\documentclass[12pt]{article}

\usepackage{amsfonts, amsmath, amssymb}
\usepackage{a4wide,color}
\usepackage{tikz,url}
\usetikzlibrary{calc,shapes}

\newcommand{\eq}{\leftrightarrow}

\newcommand{\imp}{\rightarrow}
\newcommand{\Imp}{\Rightarrow}

\newcommand{\Pmi}{\Leftarrow}
\newcommand{\et}{\wedge}
\newcommand{\vel}{\vee}

\renewcommand{\phi}{\varphi}
\newcommand{\union}{\cup}
\newcommand{\Union}{\bigcup}
\newcommand{\inter}{\cap}
\newcommand{\Inter}{\bigcap}

\newcommand{\M}{\widehat{K}}

\newcommand{\np}{\overline{p}}

\usepackage[thmmarks]{ntheorem}
\theorembodyfont{\rmfamily}
\theoremsymbol{\ensuremath{\dashv}}
\usepackage{newproof}

\newtheorem{theorem}{Theorem}
\newtheorem{example}[theorem]{Example}
\newtheorem{definition}[theorem]{Definition}
\newtheorem{proposition}[theorem]{Proposition}

\newtheorem{corollary}[theorem]{Corollary}

\newtheorem{lemma}[theorem]{Lemma}

\newcommand{\lang}{\mathcal L}
\newcommand{\model}{\mathcal M}
\newcommand{\C}{\mathcal C}
\newcommand{\VV}{\mathcal V}
\newcommand{\FF}{\mathcal F}
\newcommand{\weg}[1]{}

\newcommand{\powerset}{\mathcal{P}}

\newcommand{\lbr}{[\![}
\newcommand{\rbr}{]\!]}
\newcommand{\I}[1]{\lbr #1 \rbr}

\newcommand{\sstar}{\mathsf{star}}

\newcommand{\pure}{\mathsf{pure}}

\begin{document}

\title{Wanted Dead or Alive \\ Epistemic logic for impure simplicial complexes}

\author{Hans van Ditmarsch\thanks{University of Toulouse, CNRS, IRIT, {\tt hans.van-ditmarsch@irit.fr}} \and Roman Kuznets\thanks{TU Wien, Austria, {\tt roman.kuznets@tuwien.ac.at}}}
%\date{\today}
\date{}
\maketitle

\begin{abstract} 
We propose a logic of knowledge for impure simplicial complexes. Impure simplicial complexes represent synchronous distributed systems under uncertainty over which processes are still active (are alive) and which processes have failed or crashed (are dead). Our work generalizes the logic of knowledge for pure simplicial complexes, where all processes are alive, by Goubault et al. In our semantics, given a designated face in a complex, a formula can only be true or false there if it is defined. The following are undefined: dead processes cannot know or be ignorant of any proposition, and live processes cannot know or be ignorant of factual propositions involving processes they know to be dead. The semantics are therefore three-valued, with undefined as the third value. We propose an axiomatization that is a version of the modal logic {\bf S5}. We also show that impure simplicial complexes correspond to certain Kripke models where agents' accessibility relations are equivalence relations on a subset of the domain only. This work extends conference publication \cite{Ditmarsch21}.
\end{abstract}

\section{Introduction} \label{section.introduction}

Epistemic logic investigates knowledge and belief, and change of knowledge and belief, in multi-agent systems. A foundational study is \cite{hintikka:1962}. Knowledge change was extensively modelled in temporal epistemic logics~\cite{Alur2002,halpernmoses:1990,Pnueli77,dixonetal.handbook:2015} and more recently in dynamic epistemic logics~\cite{baltagetal:1998,hvdetal.del:2007,moss.handbook:2015} including semantics based on histories of epistemic actions, both synchronously \cite{jfaketal.JPL:2009} and asynchronously \cite{degremontetal:2011}.

Combinatorial topology has been used in distributed computing to model concurrency and asynchrony since \cite{BiranMZ90,FischerLP85,luoietal:1987}, with higher-dimensional topological properties entering the scene in \cite{HS99,herlihyetal:2013}. The basic structure in combinatorial topology is the \emph{simplicial complex}, a collection of subsets called \emph{simplices} of a set of \emph{vertices}, closed under containment. Geometric manipulations such as subdivision have natural combinatorial counterparts. 

An epistemic logic interpreted on simplicial complexes has been proposed in~\cite{goubaultetal:2018,ledent:2019,goubaultetal_postdali:2021}, including exact correspondence between simplicial complexes and certain multi-agent Kripke models where all relations are equivalence relations. Also, in those works and in e.g.\ \cite{diego:2021} the action models of~\cite{baltagetal:1998} are used to model distributed computing tasks and algorithms, with asynchronous histories treated as in \cite{degremontetal:2011}. Action models also reappear in their combinatorial topological incarnations as simplicial complexes \cite{goubaultetal:2018}.

\begin{example} \label{figure.ssss}
Below are some simplicial complexes and corresponding Kripke models. These simplicial complexes are for three agents. The vertices of a simplex are required to be labelled with different agents. A simplex that is maximal under the containment relation is called a \emph{facet}. All facets below consist of $3$ vertices. 
These are the triangles in the figure. For $2$ agents we get lines/edges, for $4$ agents we get tetrahedra, etc. Such a triangle corresponds to a state in a Kripke model. A label like $0_a$ on a vertex represents that it is a vertex for agent $a$ and that agent $a$'s local state has value $0$, etcetera. We can see this as the Boolean value of a local proposition, where $0$ means false and $1$ means true. Together these labels determine the valuation in a corresponding Kripke model, for example in states labelled $0_a1_b1_c$ agent $a$'s value is $0$, $b$'s is $1$, and $c$'s is $1$. The single triangle corresponds to the singleton $\mathcal S5$ model below it. (We assume reflexivity and symmetry of accessibility relations.) With two triangles, if they only intersect in $a$ it means that agent $a$ cannot distinguish these states, so that $a$ is uncertain about the value of $b$; whereas if they intersect in $a$ and $c$ both $a$ and $c$ are uncertain about the value of $b$. 

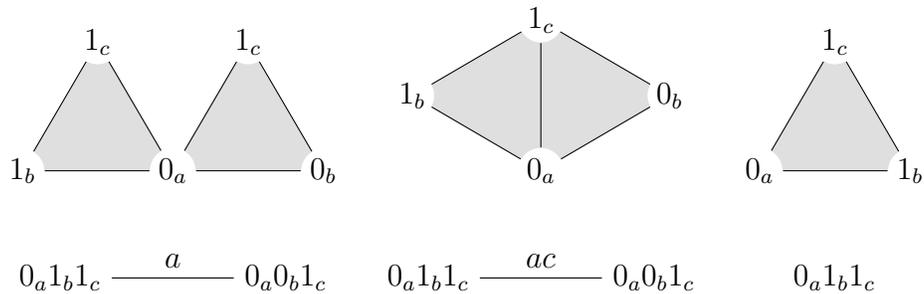
\begin{figure}[ht] \center
\scalebox{1}{
\begin{tabular}{ccc}
\begin{tikzpicture}[round/.style={circle,fill=white,inner sep=1}]
%\fill[fill=gray!25!white] (0,0) -- (2,0) -- (1,1.71) -- cycle;
\fill[fill=gray!25!white] (2,0) -- (4,0) -- (3,1.71) -- cycle;
\fill[fill=gray!25!white] (0,0) -- (2,0) -- (1,1.71) -- cycle;
\node[round] (b1) at (0,0) {$1_b$};
\node[round] (b0) at (4,0) {$0_b$};
\node[round] (c1) at (3,1.71) {$1_c$};
\node[round] (lc1) at (1,1.71) {$1_c$};
\node[round] (a0) at (2,0) {$0_a$};
%\node (f1) at (3,.65) {$Y'$};
%\node (f1) at (1,.65) {$W'$};
%
\draw[-] (b1) -- (a0);
\draw[-] (b1) -- (lc1);
\draw[-] (a0) -- (lc1);
\draw[-] (a0) -- (b0);
\draw[-] (b0) -- (c1);
\draw[-] (a0) -- (c1);
\end{tikzpicture}
&
\begin{tikzpicture}[round/.style={circle,fill=white,inner sep=1}]
%\fill[fill=gray!25!white] (0.29,1) -- (2,0) -- (2,2) -- cycle;
\fill[fill=gray!25!white] (2,0) -- (2,2) -- (3.71,1) -- cycle;
\fill[fill=gray!25!white] (0.29,1) -- (2,0) -- (2,2) -- cycle;
\node[round] (b1) at (.29,1) {$1_b$};
\node[round] (b0) at (3.71,1) {$0_b$};
\node[round] (c1) at (2,2) {$1_c$};
%\node[round] (lc1) at (1,1.71) {$0_c$};
\node[round] (a0) at (2,0) {$0_a$};
%\node (f1) at (3,.65) {$Y'$};
%\node (f1) at (1,.65) {$W'$};
%
\draw[-] (b1) -- (a0);
\draw[-] (b1) -- (c1);
%\draw[-] (a0) -- (lc1);
\draw[-] (a0) -- (b0);
\draw[-] (b0) -- (c1);
\draw[-] (a0) -- (c1);
\end{tikzpicture}
& 
\begin{tikzpicture}[round/.style={circle,fill=white,inner sep=1}]
%\fill[fill=gray!25!white] (0,0) -- (2,0) -- (1,1.71) -- cycle;
\fill[fill=gray!25!white] (2,0) -- (4,0) -- (3,1.71) -- cycle;
\node[round] (b0) at (4,0) {$1_b$};
\node[round] (c1) at (3,1.71) {$1_c$};
\node[round] (a0) at (2,0) {$0_a$};
%\node (f1) at (3,.65) {$Y$};
%
\draw[-] (a0) -- (b0);
\draw[-] (b0) -- (c1);
\draw[-] (a0) -- (c1);
\end{tikzpicture} \\ && \\
\begin{tikzpicture}
\node (010) at (.5,0) {$0_a1_b1_c$};
\node (001) at (3.5,0) {$0_a0_b1_c$};
\draw[-] (010) -- node[above] {$a$} (001);
\end{tikzpicture}
&
\begin{tikzpicture}
\node (010) at (.5,0) {$0_a1_b1_c$};
\node (001) at (3.5,0) {$0_a0_b1_c$};
\draw[-] (010) -- node[above] {$ac$} (001);
\end{tikzpicture}
&
\begin{tikzpicture}
\node (010) at (.5,0) {$0_a1_b1_c$};
\end{tikzpicture}
\end{tabular}
}
\end{figure}

The current state of the distributed system is represented by a distinguished facet of the simplicial complex, just as we need a distinguished or actual state in a Kripke model in order to evaluate propositions. For example, in the leftmost triangle, as well as in the leftmost state/world, $a$ is uncertain whether the value of $b$ is $0$ or $1$, whereas $b$ knows that its value is $1$, and all three agents know that the value of $c$ is $1$.
\end{example}

The so-called {\em impure simplicial complexes} were beyond the scope of the epistemic logic of \cite{goubaultetal:2018}. Impure simplicial complexes encode uncertainty over which processes are still active/alive. They model information states in synchronous message passing with crash failures (processes `dying').

\begin{example} \label{figure.sergio}
The impure complex below is found in \cite[Section 13.5.2]{herlihyetal:2013}, here decorated with local values in order to illustrate epistemic  features. Some vertices have been named.

\begin{figure}[ht] \center
\scalebox{1}{
\begin{tikzpicture}[round/.style={circle,fill=white,inner sep=1}]
%\fill[fill=gray!25!white] (0,0) -- (2,0) -- (1,1.71) -- cycle;
\fill[fill=gray!25!white] (2,0) -- (4,0) -- (3,1.71) -- cycle;
\node[round] (b1) at (2,-2) {$0_b$};
\node (u) at (1.6,-2.2) {$u$};
\node[round] (a1) at (4,-2) {$0_a$};
\node[round] (b0) at (4,0) {$0_b$};
\node[round] (c1) at (3,1.71) {$1_c$};
\node[round] (x) at (2.6,1.51) {$x$};
\node[round] (a0) at (2,0) {$0_a$};
\node (v) at (1.6,-.2) {$v$};
\node (v) at (4.4,-.2) {$z$};
\node[round] (lc1) at (1.3,2.71) {$0_a$};
\node[round] (la0) at (.3,1) {$1_c$};
\node (w) at (-.1,.8) {$w$};
%\node (f1) at (3,.65) {$Y$};
\node[round] (rb0) at (5.7,1) {$1_c$};
\node[round] (rc1) at (4.7,2.71) {$0_b$};
%
%\draw[-] (b1) -- node[above] {$X$} (a0);
\draw[-] (b1) -- (a0);
\draw[-] (a0) -- (b0);
\draw[-] (b0) -- (c1);
\draw[-] (a0) -- (c1);
\draw[-] (la0) -- (lc1);
\draw[-] (a0) -- (la0);
\draw[-] (lc1) -- (c1);
\draw[-] (a1) -- (b1);
\draw[-] (a1) -- (b0);
\draw[-] (b0) -- (rb0);
\draw[-] (rc1) -- (c1);
\draw[-] (rb0) -- (rc1);
\end{tikzpicture}
}
\end{figure}
This simplicial complex represents the result of possibly failed message passing between $a,b,c$. In the ($2$-dimensional) triangle in the middle the messages have all been received, whereas in the ($1$-dimensional) edges on the side a process has crashed: on the left, $b$ is dead, on the right, $a$ is dead, and below, $c$ is dead. 

We propose to interpret this complex in terms of knowledge and uncertainty. Consider triangle $vxz$ (that is $\{v,x,z\}$) and agent $a$, who `colours' (labels) vertex $v$. Vertex $v$ is the intersection of triangle $vxz$, edge $wv$, and edge $vu$. This represents $a$'s uncertainty about the actual state of information: either all three agents are alive, or $b$ is dead, or $c$ is dead. 

First, we now say, as before, that in triangle $vxz$ the value of $c$ is $1$, as in vertex $x$ the value of $c$ is $1$. Second, we now also say, which is novel, that in triangle $vxz$ agent $a$ knows that $c$'s value is $1$. This is justified because in all simplices intersecting in the $a$ vertex $v$, whenever $c$ is alive, its value is $1$. It is sufficient to consider maximal simplices (facets), where, unlike in Example~\ref{figure.ssss}, some facets consist of two vertices and other facets consist of three vertices: in edge $wv$ agent $c$'s value is $1$, in triangle $vxz$ agent $c$'s value is $1$, and in edge $vu$ agent $c$ is dead, so we are excused from considering its value. Third, the semantics we will propose also justifies that in triangle $vxz$ agent $a$ knows that $b$ knows that the value of $c$ is $1$: whenever proposition `$b$ knows that the value of $c$ is $1$' denoted $\phi$ is defined (can be interpreted), it is true: $\phi$ is undefined in $wv$ because $b$ is dead, $\phi$ is undefined in $vu$ because $b$ knows that $c$ is dead, but $\phi$ is defined in $vxz$ and (just as for $a$) true.\footnote{In fact, in the triangle $vxz$ the three agents have common knowledge ($a$ knows that $b$ knows that $c$ knows \dots) of their respective values.}

The point of evaluation can be a facet such as $vxz$ or $wv$ but it may be any simplex, such as vertex $v$. Also in vertex $v$, agent $a$ knows that $c$'s value is $1$. Such knowledge is unstable: an update of information (such as a model restriction) may remove $a$'s uncertainty and make her learn that the real information state is edge $vu$ wherein $c$ is dead. Then, it is no longer true that $a$ knows that the value of $c$ is $1$. However, a different update could have confirmed that the real information state is $vxz$, or $wv$, which would bear out $a$'s knowledge.  
\end{example}

\paragraph*{Our results.}
Inspired by the epistemic logic for pure complexes of \cite{goubaultetal:2018}, the knowledge semantics for certain Kripke models involving alive and dead agents in \cite{diego:2019}, and impure complexes modelling synchronous message passing in \cite{herlihyetal:2013}, we propose an epistemic logic for impure complexes. We define a three-valued modal logical semantics for an epistemic logical language, interpreted on arbitrary simplices of simplicial complexes that are decorated with agents and with local propositional variables. The issue of the definability of formulas has to be handled with care: standard notions such as validity, equivalence, and the interdefinability of dual modalities and non-primitive propositional connectives have to be properly addressed. This involves detailed proofs. As we interpret formulas in arbitrary simplices, results are shown for upwards and downwards monotony of truth that also depend on definability. For example, agent $c$ may have value $1$ in a simplex with a $c$ vertex, but the same proposition is undefined in a face of that simplex without that vertex.  We propose an axiomatization {\bf S5}$^\top$ of our logic for which we show soundness: apart from some usual {\bf S5} axioms and rules it also contains some of those in versions adapted to definability. For example, {\bf MP} (modus ponens) is invalid, but instead we have a derivation rule {\bf MP}$^\top$ stating ``from $\phi$ and $\phi\imp\psi$ infer $\phi^\top\imp\psi$'', where the last validity means that ``whenever $\phi$ and $\psi$ are both defined, $\psi$ is true''.  We then show that the epistemic semantics of pure complexes of \cite{goubaultetal:2018} is a special case of our semantics, and also how their version of the logic {\bf S5} is then recovered. Finally we show how impure simplicial complexes correspond to certain Kripke models with special conditions on the relations and valuations. Such Kripke models have a surplus of information that is absent in their simplicial correspondents. In Kripke models all propositional variables `must' have a value, even those of agents that are dead. But not in complexes. 

Some issues are left for further research, notably: the completeness of the axiomatization, the proper notion of bisimulation for impure complexes, the explicit representation of life and death in the language, and the extension of the logical language with dynamic modalities representing distributed algorithms and other updates.

\paragraph*{Relevance for distributed computing.}
We hope that our research is
relevant for modelling knowledge both in synchronous and
asynchronous distributed computing. Indeed, simplicial complexes are a
very versatile alternative for Kripke models where all agents know their
own state, as they make the powerful theorems of
combinatorial topology available for epistemic analysis \cite{herlihyetal:2013}.
Whereas this was already known for pure complexes \cite{HS99}, our
work reveals that it is also true for impure complexes. Indeed,
crashed processes can alternatively be modelled as non-responding
processes in asynchronous message passing, on (larger) pure simplicial
complexes. Dually, non-responding processes can be modelled as dead
processes in a synchronous setting, for example after a time-out, on
(smaller) impure complexes. Impure simplicial complexes therefore
appear as a useful abstraction even for modelling asynchronous
computations.

\paragraph*{Survey of related works.}
We will compare our work in technical detail in a (therefore) later section to various related works. %
Uncertainty over which agents are alive and dead relates to epistemic logics of \emph{awareness} (of other agents) \cite{faginetal:1988,halpernR13,AgotnesA14a,hvdetal.jolli:2014}. %
Our modal semantics with the three values true, false and undefined relates to other \emph{multi-valued epistemic logics} \cite{Morikawa89,Fitting:1992,odintsovetal:2010,RivieccioJJ17}. %
Our epistemic notion that comes `just short of knowledge' is different from various notions of {\em belief} \cite{hintikka:1962,stalnaker:2005} including such notions for distributed systems \cite{MosesS93,HalpernSS09a}, and in particular from notions of \emph{knowledge} recently proposed in \cite{diego:2019} and in \cite{goubaultetal:2021}: these are also interpreted on impure simplicial complexes and therefore of particular interest. Finally, a \emph{temporal} modal logic interpreted on simplicial complexes has been proposed in \cite{loretietal:2023}.

\paragraph*{Outline of our contribution.}
Section~\ref{section.preliminaries} gives an introduction to simplicial complexes. Section~\ref{section.el} presents the epistemic semantics for impure complexes. Section~\ref{section.validities} presents the axiomatization {\bf S5}$^\top$ for our epistemic semantics. Section~\ref{section.pure} shows that the special case of pure complexes returns the prior semantics and logic of \cite{goubaultetal:2018}.
Section \ref{section.correspondence} transforms impure complexes into a certain kind of Kripke models and vice versa. Section \ref{section.further} compares our results to the literature. 
Section~\ref{section.fff} describes the above topics for further research in detail.

\section{Technical preliminaries: simplicial complexes}  \label{section.preliminaries}

Given are a finite set $A$ of \emph{agents} (or \emph{processes}, or \emph{colours}) $a_0,a_1,\dots$ (or $a,b,\dots$) and a countable set $P = \Union_{a \in A} P_a$ of (\emph{propositional}) \emph{variables}, where all $P_a$ are mutually disjoint. The elements of $P_a$ for some $a \in A$ are the {\em local variables for agent $a$} and are denoted $p_a, q_a, p'_a, q'_a, \dots$ Arbitrary elements of $P$ are also denoted $p,q,p',q',\dots$ As usual in combinatorial topology, the number $|A|$ of agents is taken to be $n+1$ for some $n \in \mathbb N$, so that the dimension of a simplicial complex (to be defined below), that is one less than the number of agents, is $n$.

%\paragraph*{Language.}
\begin{definition}[Language]
The {\em language of epistemic logic} $\lang_{K}(A,P)$ is defined as follows, where $a \in A$ and $p_a \in P_a$. \[ \phi ::= p_a \mid \neg\phi \mid (\phi\et\phi) \mid \M_a \phi \]  
\end{definition}
Parentheses will be omitted unless confusion results. We will write $\lang_{K}(P)$ if $A$ is clear from the context, and $\lang_{K}$ if $A$ and $P$ are clear from the context. Connectives $\imp$, $\eq$, and $\vel$ are defined by abbreviation as usual, as well as $K_a\phi := \neg \M_a \neg \phi$. In the semantics and inductive proofs of our work, $\M_a\phi$ is a more suitable linguistic primitive than the more common $K_a\phi$. Otherwise, we emphasize that the choice whether $K_a$ or $\M_a$ is syntactic primitive is pure syntactic sugar. For $\emptyset \neq B\subseteq A$, we write $\lang_{K}|B$ for $\lang_K(B,\Union_{a \in B} P_a)$, and where $\lang_K|b$ means $\lang_K|\{b\}$. For $\neg p$ we may write $\np$. Expression $K_a \phi$ stands for `agent $a$ knows (that) $\phi$' and $\M_a \phi$ stands for `agent $a$ considers it possible that $\phi$' that we will abbreviate in English as `agent $a$ considers (that) $\phi$.' The fragment without inductive construct $\M_a\phi$ is the {\em language of propositional logic} (or {\em Booleans}) denoted $\lang_\emptyset(A,P)$. 

%\paragraph*{Simplicial complexes.}
\begin{definition}[Simplicial complex]
Given a non-empty set of \emph{vertices} $V$ (or {\em local states}; singular form {\em vertex}), a \emph{(simplicial) complex} $C$ is a set of non-empty finite subsets of $V$, called \emph{simplices}\footnote{Privileged plural forms are: vertices, simplices, and complexes, where the English plural is privileged as complex is a common modern term, see \cite{herlihyetal:2013}.}, that is closed under non-empty subsets (such that for all $X \in C$, $\emptyset \neq Y \subseteq X$ implies $Y \in C$), and that contains all singleton subsets of $V$.
\end{definition} If $Y \subseteq X$ we say that $Y$ is a \emph{face} of $X$. A maximal simplex in $C$ is a \emph{facet}. The facets of a complex $C$ are denoted as $\FF(C)$, and the vertices of a complex $C$ are denoted as $\VV(C)$. The \emph{star} of $X$, denoted $\sstar(X)$, is defined as $\{ Y \in C \mid X \subseteq Y\}$, where for $\sstar(\{v\})$ we write $\sstar(v)$. The dimension of a simplex $X$ is $|X|-1$, e.g., vertices are of dimension $0$, while edges are of dimension $1$. The dimension of a complex is the maximal dimension of its facets. A simplicial complex is \emph{pure} if all facets have the same dimension. Otherwise it is \emph{impure}. 

%A {\em manifold} is a pure simplicial complex $C$ of dimension $n$ such that: (i) for any $X,Y \in \FF(C)$ there are facets $X = X_0, \dots, X_m = Y$ such that for all $i < m$ the dimension of $X_i \inter X_{i+1}$ is $n-1$ (this is a `path' from $X$ to $Y$ in $C$), and (ii) any simplex $Z \in C$ of dimension $n-1$ is a face of one or two facets of $C$ (those that are faces of only one facet form the {\em boundary} of $C$).

Complex $D$ is a \emph{subcomplex} of complex $C$ if $D \subseteq C$. Given processes $B \subseteq A$, a subcomplex of interest is the \emph{$m$-skeleton} $D$ of an $n$-dimensional complex $C$, that is, the maximal subcomplex $D$ of $C$ of dimension $m < n$, that can be defined as $D := \{ X \in C \mid |X| \leq m+1\}$. We will use this term for pure and impure complexes, where we typically consider a pure $m$-skeleton of an impure $n$-dimensional complex.

We decorate the vertices of simplicial complexes with agents' names, that we often refer to as  \emph{colours}. A \emph{chromatic map} $\chi: \VV(C) \imp A$ assigns colours to vertices such that different vertices of the same simplex are assigned different colours.  Thus, $\chi(v)=a$ denotes that the local state or vertex $v$ belongs to agent $a$. Dually, the vertex of a simplex~$X$ coloured with $a$ is denoted $X_a$. Expression $\chi(X)$, for $X \in C$, denotes $\{\chi(v) \mid v \in X\}$. A pair $(C,\chi)$ consisting of a simplicial complex $C$ and a chromatic map $\chi$ is a \emph{chromatic simplicial complex}. From now on, all simplicial complexes will be chromatic simplicial complexes. 

We extend the usage of the term `skeleton' as follows to chromatic simplicial complexes. Given processes $B \subseteq A$, the $B$-skeleton of a chromatic complex $(C,\chi)$, denoted $(C,\chi)|B$, is defined as $\{ X \in C \mid \chi(X) \subseteq B \}$ (it is required to be non-empty).

\paragraph*{Simplicial models.}
We decorate the vertices of simplicial complexes  with local variables $p_a \in P_a$ for $a \in A$, where we recall that $\Union_{a \in A} P_a = P$.  \emph{Valuations} (valuation functions) assigning sets of local variables for agents $a$ to vertices coloured $a$ are denoted $\ell, \ell', \dots$ Given a vertex $v$ coloured $a$, the set $\ell(v) \subseteq P_a$ consists of $a$'s local variables that are true at $v$, where those in $P_a\setminus\ell(v)$ are the variables that are false in $a$. 
 For any $X \in C$, $\ell(X)$ stands for $\Union_{v \in X} \ell(v)$. 
\begin{definition}[Simplicial model]
A \emph{simplicial model} $\C$ is a triple $(C,\chi,\ell)$ where $(C,\chi)$ is a chromatic simplicial complex and $\ell$ a valuation function, and a {\em pointed simplicial model} is a pair $(\C,X)$ where $X \in C$. \end{definition} We slightly abuse the language by allowing terminology for simplicial complexes to apply to simplicial models, for example, $\C$ is impure if $C$ is impure, $\C|B$ is a $B$-skeleton if $(C,\chi)|B$ is a $B$-skeleton, etcetera.  A pointed simplicial model is also called a simplicial model.

 \weg{ For $\C'=(C',\chi',\ell')$ to be a subdivision of $\C=(C,\chi,\ell)$, we additionally require that for all $(v,X) \in \VV(C')$ with $\chi(v)=a$, $p_a \in\ell'(v,X)$ iff $p_a\in \ell(v)$.\footnote{For an epistemic logician, the chromatic subdivision is the result of an epistemic action wherein the agents inform each other of the value of their local variables. There is a corresponding action model \cite{goubaultetal:2018,ledent:2019}.}
}

\paragraph*{Simplicial maps.}

%**maybe needs to be adjusted for impure complexes**

A \emph{simplicial map} ({\em simplicial function}) between simplicial complexes $C$ and $C'$ is a simplex preserving function $f$ between its vertices, i.e., $f: \VV(C) \imp \VV(C')$ such that for all $X \in C$, $f(X) \in C'$, where $f(X) := \{ f(v) \mid v \in X \}$. We let $f(C)$ stand for $\{ f(X) \mid X \in C \}$. A simplicial map is {\em rigid} if it is dimension preserving, i.e., if for all $X \in C$, $|f(X)| = |X|$. We will abuse the language and also call $f$ a simplicial map between  chromatic simplicial complexes $(C,\chi)$ and $(C',\chi')$, and between simplicial models $(C,\chi,\ell)$ and $(C',\chi',\ell')$.

A \emph{chromatic simplicial map} is a {\em colour preserving} simplicial map between chromatic simplicial complexes, i.e., for all $v \in \VV(C)$, $\chi'(f(v)) = \chi(v)$. We note that it is therefore also rigid. A simplicial map $f$ between simplicial models is {\em value preserving} if for all $v \in \VV(C)$, $\ell'(f(v)) = \ell(v)$. If $f$ is not only colour preserving but also value preserving, and its inverse $f^{-1}$ as well, then~$\C$ and~$\C'$ are \emph{isomorphic}, notation ${\C \simeq \C'}$.  

\paragraph*{Epistemic logic on pure simplicial models.}

In \cite{ledent:2019,goubaultetal:2018} an epistemic semantics is given on pure simplicial models with crucial clause that $K_a \phi$ ($\M_a\phi$) is true in a facet $X$ of a pure simplicial model $\C = (C,\chi,\ell)$ of dimension $|A|-1$, if $\phi$ is true in all facets (some facet) $Y$ of $\C$ such that $a \in \chi(X \inter Y)$.
They then proceed by showing that {\bf S5} augmented with `locality' axioms $K_a p_a \vel K_a \neg p_a$ (formalizing that all agents know their local variables) is the epistemic logic of simplicial complexes. In \cite{ledent:2019,goubaultetal_postdali:2021,hvdetal.simpl:2022} bisimulation for simplicial complexes is proposed and the required correspondence (on finite models) shown. The subsection entitled `A semantics for simplices including facets' of  \cite{hvdetal.simpl:2022} proposes an alternative semantics for pure complexes wherein formulas can be interpreted in any face of a facet and not merely in facets. Somewhat surprisingly, this proposal applies to impure complexes as well, with minor adjustments. We will proceed with such an epistemic semantics based on impure complexes.

\section{Epistemic logic on impure simplicial models}  \label{section.el}

\subsection{A logical semantics relating definability and truth}

Given agents $A$ and variables $P = \Union_{a \in A} P_a$ as before, let $\C = (C,\chi,\ell)$ be a simplicial model, $X \in C$, and $\phi \in \lang_K(A,P)$. Informally, we now wish to define a satisfaction relation $\models$ between {\bf some} but not all pairs $(\C,X)$ and formulas $\phi$. Not all, because if agent $a$ does not occur in $X$ (if $a \notin \chi(X)$), we do not wish to interpret certain formulas involving $a$, such as $p_a$ and formulas of shape $K_a \phi$. This is, because, if $X$ were a facet (a maximal simplex), the absence of $a$ would mean that the process is absent/dead, and dead processes do not have local values or know anything. The relation should therefore be partial. However, this relation is fairly complex, because we may wish to interpret formulas $K_b \phi$ in $(\C,X)$, with $b 
\in\chi(X)$, where after all agent $a$ `occurs' in $\phi$, for example expressing that a `live' process $b$ is uncertain whether process $a$ is `dead'. Formally, we therefore proceed slightly differently. 

We first define an auxiliary relation $\bowtie$ between (all) pairs $(\C,X)$ and formulas $\phi$, where $\C,X \bowtie \phi$ informally means that (the interpretation of) $\phi$ {\em is defined in} $(\C,X)$. The parentheses in $(\C,X)$ are omitted for notational brevity. For ``not $\C,X\bowtie \phi$'' we write $\C,X \not\bowtie\phi$, for ``$\phi$ is undefined in $(\C,X)$''. Subsequently we then formally also define relation $\models$ between (after all) all pairs $(\C,X)$ and formulas $\phi$, where, as usual, $\C,X \models \phi$ means that $\phi$ {\em is true in} $(\C,X)$, and $\C,X \models \neg\phi$ means that $\phi$ {\em is false in} $(\C,X)$. Again, we omit the parentheses in $(\C,X)$ for notational brevity. For ``not $\C,X \models \phi$'' we write $\C,X \not\models \phi$. Unusually, $\C,X \not\models \phi$ does not mean that $\phi$ is false in $(\C,X)$ but only means that $\phi$ is not true in $(\C,X)$, in which case it can be either false in $(\C,X)$ or undefined in $(\C,X)$. 
\begin{definition}[Definability and satisfaction relation] \label{def.defsat}
We define the definability relation $\bowtie$ and subsequently the satisfaction relation $\models$ by induction on $\phi\in \lang_K$.
\[ \begin{array}{lcl}
%\C, X \bowtie \bot & \text{iff} & \text{true} \\
\C, X \bowtie p_a & \text{iff} & a \in \chi(X) \\
\C, X \bowtie \phi\et\psi & \text{iff} & \C, X \bowtie \phi \ \text{and} \ \C, X \bowtie \psi \\
\C, X \bowtie \neg \phi & \text{iff} & \C, X \bowtie \phi \\
\C,X \bowtie \M_a\phi & \text{iff} & \C,Y \bowtie \phi \ \text{for some} \ Y \in C \ \text{with} \ a \in \chi(X \inter Y) 
\end{array}\]
%Given $\bowtie$, we can define the satisfaction relation $\models$ by induction on $\phi$. %(In the knowledge clause for $\bowtie$ it says `for some' and in the knowledge clause for $\models$ it says `for all'. The difference is crucial.) 
\[ \begin{array}{lcl}
%\C, X \models \bot & \text{iff} & \text{false} \\
\C, X \models p_a & \text{iff} & a \in \chi(X) \ \text{and} \ p_a \in \ell(X) \\
\C, X \models \phi\et\psi & \text{iff} & \C, X \models \phi \ \text{and} \ \C, X \models \psi \\
\C, X \models \neg \phi & \text{iff} & \C, X \bowtie \phi \ \text{and} \ \C, X \not\models \phi \\
\C,X \models \M_a\phi & \text{iff} & \C,Y \models \phi \ \text{for some} \ Y \in C \ \text{with} \ a \in \chi(X \inter Y)
%\C,X \models K_a\phi & \text{iff} & \C,X \bowtie K_a\phi \ \text{and} \\ && \C,Y \bowtie \phi \text{ implies } \C,Y \models \phi \ \text{for all} \ Y \in C \ \text{with} \ a \in \chi(X \inter Y)
\end{array}\]
Given $\phi,\psi\in\lang_K$, $\phi$ is {\em equivalent} to $\psi$ if for all $(\C,X)$: \[\begin{array}{lll} \C,X \models \phi &\text{ iff }& \C,X \models \psi; \\ \C,X \models \neg\phi &\text{ iff }& \C,X \models \neg\psi; \\ \C,X \not\bowtie\phi &\text{ iff }& \C,X\not\bowtie\psi. \end{array}\] A formula $\phi\in\lang_K$ is {\em valid} if for all $(\C,X)$: $\C,X \bowtie \phi$ implies $\C,X \models \phi$.
\end{definition}
The definition of equivalence is the obvious one for a three-valued logic with values true, false, and undefined. Concerning validity, it should be clear that it could not have been defined as ``$\phi\in\lang_K$ is valid iff for all $(\C,X)$: $\C,X \models \phi$'' because then there would be no validities at all (if there is more than one agent), as even very simple formulas like $p_a \vel \neg p_a$ are undefined when interterpreted in a vertex for an agent $b \neq a$.

An equivalent formulation of the semantics for the epistemic modality for vertices is:
\[ \begin{array}{lcl}
\C,v \models \M_a\phi & \text{iff} & \C,Y \models \phi \ \text{for some} \ Y \in \sstar(v). \\ % with \\ && $\C,X \bowtie \phi$ and $\C,X \bowtie \psi$;
\end{array}\]

In this three-valued semantics $\C,X \not \models \phi$ is {\bf not} equivalent to $\C,X \models \neg \phi$. In particular, if $a$ is an agent not colouring a vertex in $X$ ($a \notin \chi(X)$), then: \[ \begin{array}{lllll} \text{for all} \ p_a \in P_a:  \C,X \not \models p_a \ \text{and} \ \C,X \not \models \neg p_a; \\ \text{for all} \ \psi \in \lang_K:  \C,X \not \models \M_a \psi \ \text{and} \ \C,X \not \models \neg \M_a \psi.
\end{array} \]
On the other hand, an agent may consider it possible that some proposition is true even when this proposition cannot be evaluated. That is, we may have (see Example~\ref{example.xxx}):
\[ \begin{array}{ll} \C,X \models \M_a \phi & \text{even  when} \\
\C,X \not\models \phi & \text{and} \\
\C,X \not\models \neg\phi
\end{array} \]

\begin{example} \label{example.xxx}
Consider the following impure simplicial model $\C$ for three agents $a,b,c$ with local variables respectively $p_a,p_b,p_c$. A vertex $v$ is labelled $0_a$ if $\chi(v)=a$ and $p_a \notin \ell(v)$, $1_b$ if $\chi(v)=b$ and $p_b \in \ell(v)$, etc. Some simplices have been named. 
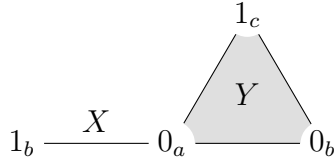
\begin{center}
\begin{tikzpicture}[round/.style={circle,fill=white,inner sep=1}]
%\fill[fill=gray!25!white] (0,0) -- (2,0) -- (1,1.71) -- cycle;
\fill[fill=gray!25!white] (2,0) -- (4,0) -- (3,1.71) -- cycle;
\node[round] (b1) at (0,0) {$1_b$};
\node[round] (b0) at (4,0) {$0_b$};
\node[round] (c1) at (3,1.71) {$1_c$};
\node[round] (a0) at (2,0) {$0_a$};
%\node[round] (bb1) at (0,-.4) {$z$};
%\node[round] (bb0) at (4,-.4) {$u$};
%\node[round] (ba0) at (2,-.4) {$w$};
\node (f1) at (3,.65) {$Y$};
\draw[-] (b1) -- node[above] {$X$} (a0);
\draw[-] (a0) -- (b0);
\draw[-] (b0) -- (c1);
\draw[-] (a0) -- (c1);
\end{tikzpicture}
\end{center}
As expected, $\C,X \models p_b \et \neg p_a$, where the conjunct $\C,X \models \neg p_a$ is justified by $\C,X \bowtie p_a$ and $\C,X \not\models p_a$. We also have $\C,X \models \M_a p_b$, because $a \in \chi(X\inter X) = \chi(X)$ and $\C,X \models p_b$.

Illustrating the novel aspects of the semantics, $\C,X \not\models p_c$, because $c \notin \chi(X)$ so that $\C,X \not\bowtie p_c$. Similarly $\C,X \not\models \neg p_c$. Also, $\C,X \not\models \M_c \neg p_a$ and similarly $\C,X \not\models \neg \M_c \neg p_a$, again because $c \notin\chi(X)$. Still, $\neg p_a$ is true throughout the model: $\C,X \models \neg p_a$ and $\C,Y \models \neg p_a$.

Although $\C,X\not\bowtie p_c$, after all $\C,X \models \M_a p_c$, because $a \in \chi(X \inter Y)$ and $\C,Y \models p_c$. Statement $\M_a p_c$ says that agent $a$ considers it possible that atom $p_c$ is true. For this to be true agent $c$ does not have to be alive in facet $X$. It is sufficient that agent $a$ considers it possible that agent $c$ is alive. % (or, in general, that an agent considers it possible that some other agent considers that possible, etcetera).

We also have $\C,Y \models K_b p_c$. This is easier to see after we have introduced the (derived) semantics for knowledge directly. We then explain in Example~\ref{example.zzz} why even $\C,X \models K_a p_c$ (not a typo). 
\end{example}

In this three-valued modal-logical setting we need to prove many intuitively expected results anew over a fair number of lemmas. We finally obtain a version of the logic {\bf S5}. 

\begin{lemma} \label{lemma.modelsbowtie}
If $\C,X \models \phi$ then $\C,X \bowtie \phi$.
\end{lemma}
\begin{proof}
This is shown by induction on $\phi$. Let $(\C,X)$ be given, where $\C = (C,\chi,\ell)$.
\begin{itemize}
%\item Because $\C,X \not\models \bot$, it is true {\em ex falso} that $\C,X \models \bot$ implies $\C,X \bowtie \bot$.

\item Let $\C,X \models p_a$. Then $a \in \chi(X)$. Therefore $\C,X \bowtie p_a$.

\item $\C,X \models \neg\phi$ implies $\C,X \bowtie \phi$ by definition of the semantics. As $\C,X \bowtie \phi$ iff $\C,X \bowtie \neg\phi$ by definition of $\bowtie$, it follows that $\C,X \bowtie \neg\phi$.

\item $\C,X \models \phi\et\psi$, iff $\C,X \models \phi$ and $\C,X \models \psi$. Therefore, by induction for $\phi$ and for $\psi$, $\C,X \bowtie \phi$ and $\C,X \bowtie \psi$, which is equivalent by definition to $\C,X \bowtie \phi\et\psi$.

\item $\C,X \models \M_a\phi$ implies $\C,Y \models \phi$ for some $Y \in C$ with $a \in \chi(X \inter Y)$, which implies (by induction) $\C,Y \bowtie \phi$ for some $Y \in C$ with $a \in \chi(X \inter Y)$, iff (by definition) $\C,X \bowtie \M_a\phi$.

%\item $\C,X \models K_a\phi$ implies $\C,X \bowtie K_a\phi$ by definition of the semantics.

\end{itemize}
\vspace{-.8cm}
\end{proof}

In this semantics, $\C,X \bowtie \phi$ does not imply that all agents occurring in $\phi$ also occur in $X$. We recall Example~\ref{example.xxx} wherein $\C,X \models \M_a p_c$ even though $c \notin\chi(X)$. On the other hand, if all agents occurring in $\phi$ also occur in $X$, then $\phi$ can be interpreted (is defined) in $X$.

Let $A: \lang_K \imp \powerset(A)$ map each formula to the subset of agents occurring in the formula: $A(p_a) := \{a\}$, $A(\phi\et\psi) := A(\phi)\union A(\psi)$, $A(\neg \phi) := A(\phi)$, and $A(\M_a\phi) := A(\phi) \union \{a\}$; $A(\phi)$ is denoted as $A_\phi$. % We have slightly abused notation by having $A$ stand for the set of agents/colours as well as for the function defined here, in order to make it easier for the reader to think of $A_\phi$ as a set of agents.

\begin{lemma} \label{lemma.aphibowtie}
Let $\phi\in\lang_K$ and $(\C,X)$ be given. Then $A_\phi \subseteq \chi(X)$ implies $\C,X \bowtie \phi$.
\end{lemma}
\begin{proof}
By induction on $\phi\in\lang_K$. If $\phi = \M_a\psi$, also $A_\psi \subseteq \chi(X)$, so by induction we get $\C,X \bowtie \psi$, and with $a \in \chi(X \inter X)$ it follows that $\C,X \bowtie \M_a\psi$. All other cases are trivial.
\end{proof}
The following may be obvious, but worthwhile to emphasize.
\begin{lemma} Let $\phi\in\lang_K$, $a \in A$, and $(\C,X)$ be given. Then $\C,X\bowtie \M_a \phi$ iff $\C,X \bowtie K_a \phi$.
\end{lemma} 
\begin{proof}
We use the definition by abbreviation of $K_a\phi$: $\C,X \bowtie K_a \phi$, iff $\C,X \bowtie \neg \M_a \neg \phi$, iff $\C,X \bowtie \M_a \neg \phi$, iff $\C,Y \bowtie \neg \phi$ for some $Y$ with $a \in \chi(X \inter Y)$, iff $\C,Y \bowtie \phi$ for some $Y$ with $a \in \chi(X \inter Y)$, iff $\C,X\bowtie \M_a \phi$.
\end{proof}

\begin{lemma} \label{lemma.defidual}
Let $\C,X\bowtie\phi$. Then $\C,X\not\models\phi$ iff $\C,X\models\neg\phi$.
\end{lemma}
\begin{proof}
From right to left follows from the semantics of negation. From left to right we proceed by contraposition. Assume $\C,X\not\models\neg\phi$. Then not ($\C,X\bowtie\phi$ and $\C,X  \not\models\phi$), that is, $\C,X\not\bowtie\phi$ or $\C,X \models\phi$. As $\C,X\bowtie\phi$ , it follows that $\C,X\models\phi$.
\end{proof}
\begin{corollary} \label{corollary.bowtiemodels}
$\C,X \bowtie \phi$, iff $\C,X\models\phi$ or $\C,X\models\neg\phi$.
\end{corollary}
In other words, if $\phi$ is defined in $(\C,X)$ then it has a truth value there, and vice versa. This may help to justify the definition of semantic equivalence of formulas. %, that seemed somewhat `over the top'.
\begin{lemma} Let $\phi,\psi\in\lang_K$. Equivalent formulations of `$\phi$ is equivalent to $\psi$' are: 
\begin{enumerate}
\item For all $(\C,X)$: ($\C,X \models \phi$ iff $\C,X \models \psi$), ($\C,X \models \neg\phi$ iff $\C,X \models \neg\psi$), and ($\C,X \not\bowtie\phi$ iff $\C,X\not\bowtie\psi$). 
\item For all $(\C,X)$: ($\C,X \bowtie\phi$ iff $\C,X\bowtie\psi$), and ($\C,X \bowtie\phi$ and $\C,X\bowtie\psi$) imply ($\C,X \models \phi$ iff $\C,X \models \psi$).
\end{enumerate}\vspace{-.8cm}
\end{lemma}
\begin{proof} This easily follows from propositional reasoning and the observation that $\C,X\models\phi$, $\C,X\models\neg\phi$, and $\C,X\not\bowtie\phi$ are mutually exclusive (see also Lemma~\ref{lemma.defidual} and Corollary~\ref{corollary.bowtiemodels}.)
\end{proof}
We will sometimes use the latter formulation of equivalence instead of the former. Neither   is equivalent to ``For all $(\C,X)$: $\C,X \models \phi$ iff $\C,X \models \psi$.'' For example, $\phi = p_a \et \neg p_a$ and $\psi = p_b \et \neg  p_b$ would be `equivalent'  in that sense, as they are never true.
\begin{lemma} \label{lemma.notneg}
Any formula $\phi$ is equivalent to $\neg\neg\phi$.
\end{lemma}
\begin{proof}
Let $(\C,X)$ be such that $\C,X \bowtie \phi$. Then also $\C,X \bowtie \neg\phi$ as well as $\C,X \bowtie \neg\neg\phi$. Using Lemma~\ref{lemma.defidual} twice, we obtain that: $\C,X \models \neg\neg\phi$, iff $\C,X \not\models \neg\phi$, iff $\C,X \models \phi$. 
\end{proof}
Let $\xi[p/\phi]$ be uniform substitution in $\xi$ of $p$ by $\phi$ (replace all occurrences in $\xi$ of $p$ with $\phi$).
\begin{lemma} \label{anotherone}
Let $(\C,X)$ be given. 
If $\phi$ is equivalent to $\psi$, then $\C,X \bowtie \xi[p/\phi]$ iff $\C,X \bowtie \xi[p/\psi]$.
\end{lemma}
\begin{proof}
All cases of the proof by induction on $\xi$ are obvious (where the base case holds because the assumption entails that $\C,X \bowtie \phi$ iff $\C,X\bowtie \psi$) except the one for knowledge.

\bigskip

\noindent
$\C,X \bowtie (\M_a\xi)[p/\phi]$ \\
iff \\
$\C,X \bowtie \M_a(\xi[p/\phi])$ \\
iff \\
$\C,Y \bowtie \xi[p/\phi]$ for some $Y$ with $a \in \chi(X \inter Y)$ \\
iff \hfill induction \\
$\C,Y \bowtie \xi[p/\psi]$ for some $Y$ with $a \in \chi(X \inter Y)$ \\
iff  \\
$\C,X \bowtie \M_a(\xi[p/\psi])$ \\
iff  \\
$\C,X \bowtie (\M_a\xi)[p/\psi]$
\end{proof}

\begin{lemma} \label{lemma.seven}
If $\phi$ is equivalent to $\psi$, then $\xi[p/\phi]$ is equivalent to $\xi[p/\psi]$.
\end{lemma}
\begin{proof}
This is shown by induction on $\xi$. We show the entire proof, but only the epistemic case is of interest.

The cases for propositional variables follow directly.
\begin{itemize}
%\item $\bot[p/\phi]$ and $\bot[p/\psi]$ are both false.
\item $p[p/\phi] = \phi$ is equivalent to $p[p/\psi] = \psi$.
\item for $q\neq p$, $q[p/\phi] = q$ is equivalent to $q[p/\psi] = q$.
\end{itemize}
For the remaining cases, let $(\C,X)$ be given.  
From the assumption and Lemma~\ref{anotherone} we obtain that $\C,X \bowtie \xi[p/\phi]$ iff $\C,X \bowtie \xi[p/\psi]$. Let us therefore assume $\C,X \bowtie \xi[p/\phi]$ and $\C,X \bowtie \xi[p/\psi]$. It remains to prove that $\C,X \models \xi[p/\phi]$ iff $\C,X \models \xi[p/\psi]$.

\begin{itemize}
\item On the definability assumption we have that (Lemma~\ref{lemma.defidual}) $\C,X\models\neg\xi[p/\phi]$ iff $\C,X\not\models\xi[p/\phi]$, and also $\C,X\models\neg\xi[p/\psi]$ iff $\C,X\not\models\xi[p/\psi]$. Therefore:

\bigskip

\noindent
$\C,X\models\neg\xi[p/\phi]$ iff $\C,X\models\neg\xi[p/\psi]$, \\ iff \\
$\C,X\not\models\xi[p/\phi]$ iff $\C,X\not\models\xi[p/\psi]$, \\ iff \\
$\C,X\models\xi[p/\phi]$ iff $\C,X\models\xi[p/\psi]$.

\bigskip

\noindent
The last holds by inductive assumption.
\item $\C,X\models (\xi\et\xi')[p/\phi]$, iff  $\C,X\models \xi[p/\phi]\et\xi'[p/\phi]$, iff $\C,X\models \xi[p/\phi]$ and $\C,X\models\xi'[p/\phi]$, iff (by induction) $\C,X\models \xi[p/\psi]$ and $\C,X\models\xi'[p/\psi]$, (...) iff $\C,X\models (\xi\et\xi')[p/\psi]$.

\item $\C,X\models (\M_a\xi)[p/\phi]$, iff $\C,X\models \M_a(\xi[p/\phi])$, iff $\C,Y\models \xi[p/\phi]$ for some $Y \in C$ with $a \in \chi(X \inter Y)$. 

We now use the inductive hypothesis that $\xi[p/\phi]$ is equivalent to $\xi[p/\psi]$. Using Lemma~\ref{anotherone} again but now for $(\C,Y)$ we obtain that $\C,Y \bowtie \xi[p/\phi]$ iff $\C,Y \bowtie \xi[p/\psi]$. Also, from $\C,Y\models \xi[p/\phi]$ and Lemma~\ref{lemma.modelsbowtie} it follows that  $\C,Y\bowtie \xi[p/\phi]$. Therefore assuming $\C,Y \bowtie \xi[p/\phi]$ and $\C,Y \bowtie \xi[p/\psi]$ we may conclude that $\C,Y \models \xi[p/\phi]$ iff $\C,Y \models \xi[p/\psi]$. So that, continuing the proof:

$\C,Y\models \xi[p/\phi]$ for some $Y \in C$ with $a \in \chi(X \inter Y)$, iff
  $\C,Y\models \xi[p/\psi]$ for some $Y \in C$ with $a \in \chi(X \inter Y)$, iff (...) $\C,X\models (\M_a\xi)[p/\psi]$.
\end{itemize}
\vspace{-.8cm}
\end{proof}

Because other logical connectives are defined by notational abbreviation, we have to prove that the implied truth conditions for those other connectives still correspond to our intuitions. This is indeed the case, however, on the strict condition that both constituents of binary connectives are defined.
\begin{lemma} \label{lemma.diamond}
Let $\C = (C, \chi,\ell)$, $X \in C$, and $\phi,\psi \in \lang_K$ be given. Then:
\[ \begin{array}{llll}
\C,X \models \phi \vel \psi & \text{iff} & \C,X \bowtie \phi, \C,X \bowtie \psi, \ \text{and } \C,X \models \phi \ \text{or} \ \C,X \models \psi \\
\C,X \models \phi \imp \psi & \text{iff} & \C,X \bowtie \phi, \C,X \bowtie \psi, \ \text{and } \C,X \models \phi \ \text{implies} \ \C,X \models \psi \\
\C,X \models \phi \eq \psi& \text{iff} &\C,X \bowtie \phi, \C,X \bowtie \psi, \text{ and } \C,X \models \phi \ \text{iff} \ \C,X \models \psi \\
\C,X \models K_a\phi & \text{iff} & \C,X \bowtie K_a\phi, \ \text{and} \\ && \C,Y \bowtie \phi \text{ implies } \C,Y \models \phi \ \text{for all} \ Y \in C \ \text{with} \ a \in \chi(X \inter Y)
\end{array}\]
\end{lemma}
\begin{proof} The proofs are straightforward, but necessary, where instead of the notational abbrevations we use their definitions.
\begin{itemize}
\item Firstly, $\C,X \models \neg(\neg\phi \et \neg\psi)$, iff $\C,X \bowtie \neg\phi \et \neg\psi$ and $\C,X \not\models \neg\phi \et \neg\psi$. Then, $\C,X \bowtie \neg\phi \et \neg\psi$, iff, respectively, $\C,X \bowtie \neg\phi$ and $\C,X \bowtie \neg\psi$, iff $\C,X \bowtie \phi$ and $\C,X \bowtie \psi$. Also, $\C,X \not\models \neg\phi \et \neg\psi$, iff $\C,X \not\models \neg\phi$ or $\C,X \not\models \neg\psi$. Now using that $\C,X \bowtie \phi$ and $\C,X \bowtie \psi$, and Lemma~\ref{lemma.defidual}, $\C,X \not\models \neg\phi$ or $\C,X \not\models \neg\psi$, iff $\C,X \models \phi$ or $\C,X \models \psi$.

\item Similar to the first item.

\item The third item can be obtained by seeing $\phi \eq \psi$ as the abbreviation of $(\phi\imp\psi)\et(\psi\imp\phi)$, and then proceeding as before.

\item By definition of the semantics for negation, $\C,X \models \neg \M_a \neg \phi$ iff $\C,X \bowtie \M_a \neg \phi$ and $\C,X \not \models \M_a \neg\phi$. Then, $\C,X \bowtie \M_a \neg \phi$, iff $\C,X \bowtie \neg \M_a \neg \phi$, iff (by definition of the abbreviation) $\C,X \bowtie K_a \phi$, which establishes half of the proof obligation. Now:

\medskip

\noindent
$\C,X \not \models \M_a \neg\phi \\
\text{iff} \\
\text{not}[ \ \C,Y \models \neg\phi \text{ for some } Y \text{ with } a \in \chi(X \inter Y) \ ] \\
\text{iff} \\
\C,Y \not\models \neg\phi \text{ for all } Y \text{ with } a \in \chi(X \inter Y) \\
\text{iff} \\
\text{not}[ \ \C,Y \bowtie \phi \text{ and } \C,Y \not\models\phi \ ] \text{ for all } Y \text{ with } a \in \chi(X \inter Y) \\
\text{iff} \\
\C,Y \bowtie \phi \text{ implies } \C,Y \models\phi \text{ for all } Y \text{ with } a \in \chi(X \inter Y)$

\medskip

\noindent That establishes the other half of the proof obligation.
\end{itemize}
\vspace{-.8cm}
\end{proof}
The direct semantics for other propositional and modal connectives of Lemma~\ref{lemma.diamond} is useful in the formulation of other results and in the proofs in the continuation, but also demonstrates our preference for the chosen linguistic primitives. 

In the semantics of conjunction we need not be explicit about definability. But in the semantics for the other binary connectives we need to be explicit about definability. A formula $\phi\eq\psi$ that is an equivalence may well be true and even valid even when $\phi$ is not equivalent to $\psi$. For example, $(p_a \vel \neg p_a) \eq (p_b \vel \neg p_b)$ is valid (whenever a simplex contains vertices for $a$ and for $b$, both are true) but $p_a \vel \neg p_a$ is not equivalent to $p_b \vel \neg p_b$ (given a simplex containing a vertex for $a$ but not for $b$, the first is defined but not the second).

It is easy to see for that all $\phi \in \lang_K$, $\phi\vel\neg\phi$ is valid. We only need to consider $(\C,X)$ such that $\C,X\bowtie \phi$, so that $\phi$ has a value in $(\C,X)$: either $\C,X\models\phi$ or $\C,X\models\neg\phi$, and therefore, with Lemma~\ref{lemma.diamond}, $\C,X \models \phi\vel\neg\phi$. Still, $\C,X\not\models \phi\vel\neg\phi$ if  $\C,X\not\bowtie \phi$ (see Example~\ref{example.zzz}).

The semantics of $\M_a\phi$ we find more elegant than those of $K_a\phi$. In the  semantics of $K_a\phi$ we can replace the part $\C,X \bowtie K_a \phi$ (also repeated below) by either of the following without changing the meaning of the definition.
\[\begin{array}{lll}
(i) & \C,X \bowtie K_a \phi & \quad\quad\text{as in Lemma~\ref{lemma.diamond}} \\
(ii) & \C,Y \bowtie \phi \text{ for some } Y \text{ with } a \in \chi(X \inter Y) \\
(iii) \quad & \C,Y \models \phi \text{ for some } Y \text{ with } a \in \chi(X \inter Y)
\end{array}\]
In proofs we often use variation $(iii)$.

A consequence of the knowledge semantics is that the agent may know a proposition even if that proposition is not defined in all possible simplices. In particular the proposition may not be true in the actual simplex (although it cannot be false there): $K_a \phi$ rather means ``as far as I know, $\phi$'' than ``I know $\phi$.''  See Example~\ref{example.zzz} below. 

A good way to understand knowledge on impure complexes is dynamically: even if $K_a \phi$ is true, an update (such as a model restriction, or a subdivision) may be possible that makes agent $a$ learn that before the update $\phi$ was not true, namely when $\phi$ was not defined. Even when $K_a\phi$ is true, agent $a$ may be uncertain whether an agent involved in $\phi$ is alive, and the update can confirm that the agent was already dead.

\begin{example}\label{example.zzz}
We recall Example~\ref{example.xxx}, about the simplicial model $\C$ here depicted again.
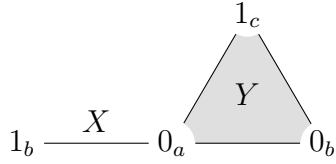
\begin{center}
\begin{tikzpicture}[round/.style={circle,fill=white,inner sep=1}]
%\fill[fill=gray!25!white] (0,0) -- (2,0) -- (1,1.71) -- cycle;
\fill[fill=gray!25!white] (2,0) -- (4,0) -- (3,1.71) -- cycle;
\node[round] (b1) at (0,0) {$1_b$};
\node[round] (b0) at (4,0) {$0_b$};
\node[round] (c1) at (3,1.71) {$1_c$};
\node[round] (a0) at (2,0) {$0_a$};
%\node[round] (bb1) at (0,-.4) {$z$};
%\node[round] (bb0) at (4,-.4) {$u$};
%\node[round] (ba0) at (2,-.4) {$w$};
\node (f1) at (3,.65) {$Y$};
\draw[-] (b1) -- node[above] {$X$} (a0);
\draw[-] (a0) -- (b0);
\draw[-] (b0) -- (c1);
\draw[-] (a0) -- (c1);
\end{tikzpicture}
\end{center}
Observe that:
\begin{itemize}
\item $\C,X \models p_b\vel\neg p_b$ (because $b \in \chi(X))$ but $\C,X\not\models p_c\vel\neg p_c$ (because $c \notin \chi(X)$).
\item $\C,X \models \M_a p_b \imp \M_a p_c$ but $\C,X \not\models p_b \imp p_c$
\item $\C,X\models K_a p_c$ but $\C,X\not\models p_c$. For $\C,X\models K_a p_c$ it suffices that $\C,Y \models p_c$, as $\C,X \not\bowtie p_c$.
\item $\C,X\not\models K_a p_c \imp p_c$, because $\C,X\not\bowtie K_a p_c \imp p_c$.
\end{itemize}
Formula $K_a p_c$ is true `by default', because given the two facets $X$ and $Y$ that agent $c$ considers possible, as far as $a$ knows, $p_c$ is true. This knowledge is defeasible because $a$ may learn that the actual facet is $X$ and not $Y$, in which case $a$ has no knowledge about agent $c$. However, $a$ considered it possible that she would have learnt that it was $Y$, in which case her knowledge was justified.

Although $\C,X\not\models K_a p_c \imp p_c$, we still have that $\models K_a p_c \imp p_c$, as for that we only need to consider $(\C',X')$ such that both $K_a p_c$ and $p_c$ are defined in $X'$. The validity of $K_a p_c \imp p_c$ means that there is no $(\C',X')$ such that $\C',X' \models \neg p_c \et K_a p_c$: knowledge cannot be verifiably incorrect.
\end{example}

\begin{example} \label{example.stackedkn}
We provide an example of stacked knowledge operators. Consider the following impure simplicial model $\C$.
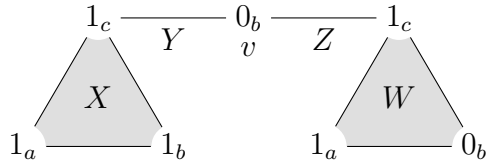
\begin{center}
\begin{tikzpicture}[round/.style={circle,fill=white,inner sep=1}]
%\fill[fill=gray!25!white] (0,0) -- (2,0) -- (1,1.71) -- cycle;
\fill[fill=gray!25!white] (2,0) -- (4,0) -- (3,1.71) -- cycle;
%\node[round] (b1) at (0,0) {$1_b$};
\node[round] (b0) at (4,0) {$1_b$};
\node[round] (c1) at (3,1.71) {$1_c$};
\node[round] (a0) at (2,0) {$1_a$};
\node (f1) at (3,.65) {$X$};
%
%\draw[-] (b1) -- node[above] {$X$} (a0);
\draw[-] (a0) -- (b0);
\draw[-] (b0) -- (c1);
\draw[-] (a0) -- (c1);

\fill[fill=gray!25!white] (6,0) -- (8,0) -- (7,1.71) -- cycle;
\node[round] (b0r) at (8,0) {$0_b$};
\node[round] (c1r) at (7,1.71) {$1_c$};
\node[round] (a0r) at (6,0) {$1_a$};
\node (f1r) at (7,.65) {$W$};
\draw[-] (a0r) -- (b0r);
\draw[-] (b0r) -- (c1r);
\draw[-] (a0r) -- (c1r);
\node[round] (b0m) at (5,1.71) {$0_b$};
\node[round] (b0mb) at (5,1.3) {$v$};
\draw[-] (c1) -- node[below] {$Y$} (b0m);
\draw[-] (c1r) -- node[below] {$Z$} (b0m);
\end{tikzpicture}
\end{center}
%Let $v$ be the vertex in the middle labeled $0_b$. 
It holds that $\C,v \models K_b K_c p_a$. This is because $v \in Y$ and $v \in Z$ and $\C,Y \models K_c p_a$ as well as $\C,Z \models K_c p_a$. In all facets agent $b$ considers possible and wherein $K_c p_a$ is defined (and it is defined in both), it is true that agent $c$ knows that the value of $a$'s local variable is $1$. But in all facets agent $b$ considers possible agent $a$ is dead, so $b$ `knows' that $a$ is dead (we cannot say this in the logical language --- see Section~\ref{section.further} for a discussion). Also, $\C,v \not\models K_c p_a$ because $c \notin \chi(v) = \{b\}$. And also, $\C,v \not\models p_a$ because $a \notin \{b\}$. And $\C,v \not\models K_b p_a$. Summing up, $K_b K_c p_a$ is true whereas none of $K_c p_a$, $K_b p_a$, and $p_a$ are true.  
\end{example}

\begin{lemma} \label{lemma.upbowtie}
If $\C,X \bowtie \phi$ and $Y \in C$ with $X \subseteq Y$, then $\C,Y \bowtie \phi$.
\end{lemma}
\begin{proof}
By induction on formula structure. Let $Y \in C$ with $X \subseteq Y$.
\begin{itemize}
%\item $\C,X \bowtie \bot$ and $\C,Y \bowtie \bot$ by definition. 
\item $\C,X \bowtie p_a$, iff $a \in \chi(X)$, which implies $a \in \chi(Y)$, iff $\C,Y \bowtie p_a$. 
\item $\C,X \bowtie \neg\phi$, iff $\C,X \bowtie \phi$, which implies (by induction) $\C,Y \bowtie \phi$, iff $\C,Y \bowtie \neg\phi$. 
\item $\C,X \bowtie \phi\et\psi$, iff $\C,X \bowtie \phi$ and $\C,X \bowtie \psi$, which implies (by induction) $\C,Y \bowtie \phi$ and $\C,Y \bowtie \psi$, iff $\C,Y \bowtie \phi\et\psi$.  
\item $\C,X \bowtie \M_a\phi$, iff $\C,Z \bowtie \phi$ for some $Z \in C$ with $a \in \chi(X \inter Z)$, which implies (as $X \subseteq Y$) $\C,Z \bowtie \phi$ for some $Z \in C$ with $a \in \chi(Y \inter Z)$, iff $\C,Y \bowtie \M_a\phi$.
\end{itemize}
\vspace{-.8cm}
\end{proof}

\begin{proposition}\label{prop.star}
If $\C,X \models \phi$ and $Y \in C$ such that $X \subseteq Y$, then $\C,Y \models \phi$.
\end{proposition}
\begin{proof}
The proof is by induction on the structure of formulas $\phi$. In order to
make the proof work for the case of negation we prove a stronger statement
(by mutual induction).
\begin{quote}
{\em Let $(\C,X)$ with $\C = (C,\chi,\ell)$ be given. For all $\phi\in\lang_K$ and for all $X,Y \in C$ with $X \subseteq Y$: $\C,X \models \phi$ implies $\C,Y \models \phi$, and $\C,X \models \neg\phi$ implies $\C,Y \models \neg\phi$.}
\end{quote}
Interestingly, the case $\M_a\phi$ below does not require the
use of the induction hypothesis.
%It will be clear that the original statement then follows immediately.
\begin{itemize}

\item $\C,X \models p_a$, iff $a \in \chi(X)$ and $p_a \in \ell(X)$, iff $a \in \chi(X)$ and $p_a \in \ell(X_a)$. As $X \subseteq Y$, also $a \in \chi(Y)$, and $X_a=Y_a$. Thus, $a \in \chi(Y)$ and $p_a \in \ell(Y_a) \subseteq \ell(Y)$ which is by definition $\C,Y \models p_a$.

\item $\C,X \models \neg p_a$, iff $\C,X \bowtie p_a$ and $\C,X \not\models p_a$, iff $a \in \chi(X)$ and $p_a \notin \ell(X)$, iff $a \in \chi(X)$ and $p_a \notin \ell(X_a)$. As $X \subseteq Y$, again we obtain that $a \in \chi(Y)$ and $X_a=Y_a$, so $p_a \notin \ell(Y_a)$. Thus, $a \in \chi(Y)$ and $p_a \notin \ell(Y)$ and therefore $\C,Y \models \neg p_a$.

\item $\C, X\models \neg \phi$ implies $\C, Y \models \neg \phi$
by the mutual part of the induction for~$\phi$.

\item $\C, X\models \neg \neg \phi$, iff (by
Lemma~\ref{lemma.notneg}) $\C, X \models \phi$, which implies by induction $\C, Y \models \phi$, iff (by once more Lemma~\ref{lemma.notneg}) $\C, Y\models
\neg \neg \phi$.

\item $\C,X \models \phi\et\psi$, iff $\C,X \models \phi$ and $\C,X \models \psi$, which implies (by induction) $\C,Y \models \phi$ and $\C,Y \models \psi$, iff $\C,Y \models \phi\et\psi$.

\item $\C,X \models \neg(\phi \et \psi)$,
iff $\C,X \bowtie \phi \et \psi$ and $\C,X \not\models \phi \et \psi$,
iff $\C,X \bowtie \phi$, $\C,X \bowtie \psi$, and $\C,X \not\models
\phi$ or $\C,X \not\models \psi$, iff $\C,X \models \neg\phi$ and $\C,X \bowtie \psi$, or  $\C,X \bowtie \phi$ and $\C,X\models \neg\psi$.
Using induction for either $\phi$ or $\psi$
and Lemma~\ref{lemma.upbowtie}, we obtain  $\C,Y \models \neg\phi$ and  $\C,Y \bowtie \psi$, or  $\C,Y \bowtie
\phi$ and $\C,Y\models \neg\psi$, which is equivalent to $\C,Y
\models \neg(\phi \et \psi)$.

\item $\C,X \models \M_a\phi$, iff $\C,Z \models \phi$ for some $Z \in C$ with $a \in \chi(X \inter Z)$, which implies (as $X \subseteq Y$) that $\C,Z \models \phi$ for some $Z \in C$ with $a \in \chi(Y \inter Z)$, iff $\C,Y \models \M_a\phi$.

\item $\C,X \models \lnot\M_a\phi$, iff $\C, X \bowtie\M_a\phi$ and $\C,X \not\models \M_a\phi$, iff $\C,Z \bowtie \phi$ for some $Z \in C$ with $a \in \chi(X \inter Z)$ and $\C,Z \not\models \phi$ for all $Z \in C$ with $a \in \chi(X \inter Z)$, iff $a \in \chi(X)$ and $\C,Z \bowtie \phi$ for some $Z \in \sstar(X_a)$  and $\C,Z \not\models \phi$ for all $Z \in \sstar(X_a)$, which implies (because $X \subseteq Y$ and $Y_a = X_a$) that $a \in \chi(Y)$ and $\C,Z \bowtie \phi$ for some $Z \in \sstar(Y_a)$  and $\C,Z \not\models \phi$ for all $Z \in \sstar(Y_a)$. This statement is equivalent to $\C,Y \models \lnot\M_a\phi$.

\end{itemize}
\vspace{-.8cm}
\end{proof}

\begin{proposition} \label{lemma.ysubx}
If $\C,X \models \phi$ and $Y \in C$ such that $Y \subseteq X$ and $\C,Y \bowtie \phi$, then $\C,Y \models \phi$.
\end{proposition}
\begin{proof} Let now $Y \in C$ such that $Y \subseteq X$. In all inductive cases we assume that the formula is defined in $Y$.
\begin{itemize}

%\item $\C,X \models \bot$ iff false.

\item $\C,X \models p_a$, iff $a \in \chi(X)$ and $p_a \in \ell(X)$, that is, $p_a \in \ell(X_a)$. As $\C,Y \bowtie p_a$, also $a \in \chi(Y)$, so that $a \in \chi(X \inter Y)$. Therefore $X_a \in X \inter Y$, so that $p_a \in \ell(X_a) \subseteq \ell(Y)$.

\item $\C,X \models \neg\phi$, iff $\C,X\bowtie\phi$ and $\C,X \not \models \phi$. Using the contrapositive of Proposition~\ref{prop.star}, $\C,X \not \models \phi$ implies $\C,Y\not \models\phi$. From that, together with the assumption $\C,Y\bowtie\phi$, we obtain by definition $\C,Y \models \neg\phi$.

\item $\C,X \models \phi\et\psi$, iff  $\C,X \models \phi$ and $\C,X \models \psi$, which implies (by induction) $\C,Y \models \phi$ and $\C,Y \models \psi$, iff $\C,Y \models \phi\et\psi$.

\item $\C,X \models \M_a\phi$, iff $\C,Z \models \phi$ for some $Z \in C$ with $a \in \chi(X \inter Z)$. Assumption $\C,Y \bowtie \M_a \phi$ implies $a \in \chi(Y)$, so that it follows from $a \in \chi(X \inter Z)$ and $Y \subseteq X$ that $a \in \chi(Y \inter Z)$. Therefore $\C,Z \models \phi$ for some $Z \in C$ with $a \in \chi(Y \inter Z)$, which is by definition $\C,Y \models \M_a\phi$.
\end{itemize}
\vspace{-.8cm}
\end{proof}
An alternative formulation of Propositions~\ref{prop.star} and \ref{lemma.ysubx} is, respectively: 
\begin{itemize}
\item If $\C,X \models \phi$ and $Y \in \sstar(X)$, then $\C,Y \models \phi$. 
\item If $\C,X \models \phi$, $\C,Y \bowtie \phi$ and $X \in \sstar(Y)$, then $\C,Y \models \phi$.
\end{itemize}
A more efficient way to determine the truth of $\phi$ may be to consider facets only. The following are straightforward consequences of the combined Propositions~\ref{prop.star} and \ref{lemma.ysubx}.
\begin{corollary} Let $\C,X \bowtie \phi$. Then all pairwise equivalent are:
\begin{itemize}
\item $\C,X \models \phi$

\item $\C,Y \models \phi$ for all faces $Y \in \sstar(X)$

\item $\C,Y \models \phi$ for all facets $Y \in \sstar(X)$
\end{itemize} \vspace{-.6cm}
\end{corollary} 

\begin{example} We recall Example~\ref{example.xxx}. Let $u$ be the vertex labelled $0_a$. Then:
\begin{itemize}
\item $\C,u \models K_a p_c$ and $\{u\} \subseteq Y$, therefore $\C,Y \models K_a p_c$ by Proposition~\ref{prop.star}.

\item $\C,Y \models p_c$ and $\{u\} \subseteq Y$, however $\C,u \not\bowtie p_c$, therefore we cannot conclude that $\C,u \models p_c$ by Proposition~\ref{lemma.ysubx}; indeed from $\C,u \not\bowtie p_c$ it follows that $\C,u \not\models p_c$.

\item $\C,Y \models \neg p_a$, $\{u\} \subseteq Y$ and $\C,u \bowtie \neg p_a$, therefore $\C,u \models \neg p_a$ by Proposition~\ref{lemma.ysubx}.
\end{itemize}\vspace{-.8cm}
\end{example}

%\subsection{Validities and differences with the logic S5} \label{section.validities}
\section{Axiomatization} \label{section.validities}

In this section we present an axiomatization for our three-valued epistemic semantics. It is a version of the well-known axiomatization {\bf S5} for the two-valued epistemic semantics. We will show soundness for this axiomatization. We will not show completeness for this axiomatization: although we are confident that it is complete, the construction of a canonical simplicial model for our three-valued semantics seems a non-trivial and lengthy technical exercise.

The axiomatization requires an auxiliary syntactic notion: for any formula $\phi$ we define an equidefinable but valid formula $\phi^\top$.
\begin{definition}
For a formula $\phi$ we define formula $\phi^\top$ recursively as follows:
\begin{itemize}
\item $p_a^\top := p_a \vel \neg p_a$;
\item $(\neg \phi)^\top := \phi^\top$;
\item $(\phi \et \psi)^\top := \phi^\top \et \psi^\top$;
\item $(\M_a \phi)^\top := \M_a \phi^\top$.
\end{itemize}
Formulas $p_a^\top$ are also denoted $\top_a$.
\end{definition}

\begin{lemma}
For any formula $\phi\in \lang_K$, $\phi^\top$ is valid.
\end{lemma}
\begin{proof}
The proof is by induction on the construction of $\phi$. The base case of $p_a$ clearly yields a valid formula. Conjunction of valid formulas is valid, and the modality $\M_a$ applied to a valid formula produces a valid formula.
\end{proof}

\begin{lemma}
\label{lem:equidef}
For any $\phi\in \lang_K$ and simplicial model $(\C,X)$: 
$\C, X \bowtie \phi$ iff $\C, X \bowtie \phi^\top$.
\end{lemma}
\begin{proof}
The proof is by induction on the construction of $\phi$. 
\end{proof}
In the continuation we will also need another, basic, lemma.
\begin{lemma} \label{lem:MPloc}
Let $(\C,X)$ and $\phi,\psi\in \lang_K$ be given. Then $\C, X \models \phi\imp \psi$ and $\C, X \models \phi$ imply $\C, X \models \psi$.
\end{lemma}
\begin{proof}
This follows directly from the derived semantics of implication (Lemma~\ref{lemma.diamond}).
%Assume that $\C, X \models \phi\imp \psi$ and $\C, X \models \phi$. It follows that $\C, X \bowtie \phi \imp \psi$, hence, both $\C, X \bowtie \phi$ and $\C, X \bowtie \psi$. Further,  $\C, X \models \phi\imp \psi$ means that either $\C, X \not\models \phi$ or $\C, X \models \psi$. The former contradicts $\C, X \models \phi$. Hence, $\C, X \models \psi$.
\end{proof}

\begin{definition}[Axiomatization S5$^\top$]
The axiomatization {\bf S5}$^\top$ consists of the following axioms and rules.
\[\begin{array}{l|l}
\mathbf{Taut} & \text{all instantiations of propositional tautologies} \\
\mathbf{L} & K_a p_a \vel K_a \neg p_a \\
\mathbf{K^\top} & K_a (\phi \imp \psi) \imp K_a \phi \imp K_a (\phi^\top\imp\psi)  \\
\mathbf{T} & K_a \phi \imp \phi  \\
\mathbf{4} & K_a \phi \imp K_a K_a \phi \\
\mathbf{5} & \M_a \phi \imp K_a \M_a \phi   \\
\mathbf{MP^\top} & \text{from } \phi\imp\psi \text{ and } \phi \text{ infer } \phi^\top \imp \psi \\
\mathbf{N} & \text{from } \phi \text{ infer } K_a \phi
\end{array}\]
\end{definition}

The soundness of this axiomatization is shown over a number of propositions. Given the presence of axiom {\bf L} featuring propositional variables $p_a$ instead of arbitrary formulas $\phi,\psi,\dots$ we can immediately observe that {\bf S5}$^\top$ is not closed under uniform substitution. For example, $K_a p_a \vel K_a \neg p_a$ is derivable (and even an axiom) but, for $a \neq b$, $K_a p_b \vel K_a \neg p_b$ is not derivable. Therefore the logic with axiomatization {\bf S5}$^\top$ is not a normal modal logic.

\begin{proposition} \label{prop.taut}
Let $\phi \in L_\emptyset(A,P)$ be a tautology. Then $\models \phi$.
\end{proposition}
\begin{proof}
Given some $\phi \in L_\emptyset$ and a simplicial model $(\C,X)$, assume $\C,X \bowtie \phi$. Then either $\C,X \models \phi$ or $\C,X \models \neg \phi$. As $\phi$ is a Boolean formula, this means that $a \in \chi(X)$ for all variables $p_a$ occurring in $\phi$, and also, that its truth value only depends on the valuation of variables $P$ in $X$. In other words, we must have that either $\ell(X)\Vdash\phi$ or $\ell(X)\Vdash\neg\phi$, where $\Vdash$ is the standard (two-valued) propositional logical satisfaction relation in $L_\emptyset|\chi(X)$. As $\phi$ is a propositional tautology, this must be $\ell(X)\Vdash\phi$, and therefore also $\C,X \models \phi$. We now have shown that $\models \phi$.
\end{proof}

\begin{proposition} \label{prop.local}
For all agents $a$ and variables $p_a$: $\models K_a p_a \vel K_a \neg p_a$. \end{proposition}
\begin{proof} Let $(\C,X)$ be given such that $\C,X \bowtie K_a p_a \vel K_a \neg p_a$. Then $\C,X \bowtie K_a p_a$ and $\C,X \bowtie K_a \neg p_a$, that is, there is a $Y$ with $a \in \chi(X \inter Y)$ such that $\C,Y \bowtie p_a$ and there is a $Z$ with $a \in \chi(X \inter Z)$ such that $\C,Z \bowtie \neg p_a$ (and therefore also  $\C,Z \bowtie p_a$). As $a \in \chi(X \inter Y)$ and $a \in \chi(X \inter Z)$ imply that $a \in \chi(X)$, so that  $\C,X \bowtie p_a$, clearly $X=Y=Z$ serves as such a witness. Therefore $\C,X \bowtie p_a$. It is easy to see that we also have that $\C,X \bowtie p_a$ implies  $\C,X \bowtie K_a p_a \vel K_a \neg p_a$. Therefore  $\C,X \bowtie K_a p_a \vel K_a \neg p_a$ iff $\C,X \bowtie p_a$.

On the assumption that $\C,X \bowtie p_a$, we now show that $\C,X \models K_a p_a \vel K_a \neg p_a$. As $\C,X \bowtie p_a$, we must have that $\C,X_a \models p_a$ or that $\C,X_a\models \neg p_a$. In the first case we obtain $\C,X_a \models K_a p_a$ and in the second case we obtain $\C,X_a \models K_a \neg p_a$ (the vertex coloured $a$ is a member of any simplex containing it). Therefore $\C,X_a \models K_a p_a \vel K_a \neg p_a$. With Proposition~\ref{prop.star} it follows that $\C,X \models K_a p_a \vel K_a \neg p_a$, as required.
\end{proof}

\begin{proposition}\label{prop.k} Let $a \in A$ and $\phi,\psi \in \lang_K$ be given. Then $\models K_a (\phi \imp \psi) \imp K_a \phi \imp K_a (\phi^\top\imp \psi)$.
\end{proposition}
\begin{proof}
Consider a simplex $X$ in a given $\C$ where the formula is defined. Then all the three $K_a$ constituents are defined. Assume that $\C, X \models K_a (\phi \imp \psi)$ and $\C, X \models K_a \phi$. Now consider any $Y$ with $a \in \chi(X \cap Y)$ such that $\C, Y \bowtie \phi^\top \imp \psi$. Note that such $Y$ must exist because $\C,X \bowtie K_a (\phi^\top \imp \psi)$ (we recall that existence is also required in the $K_a$ semantics, see Lemma~\ref{lemma.diamond} and in particular the subsequent version $(iii)$). Then we have $\C, Y \bowtie \phi^\top$ and $\C, Y \bowtie \psi$. It follows from the former by Lemma~\ref{lem:equidef} that $\C, Y \bowtie \phi$. Hence, also $\C, Y \bowtie \phi \imp \psi$. Given our assumptions, we conclude that $\C, Y \models \phi$ and $\C, Y \models \phi \imp \psi$. Hence, by Lemma~\ref{lem:MPloc}, $\C, Y \models \psi$, implying that $\C, Y \models \phi^\top \imp \psi$. Since $Y$ was arbitrary and since some such $Y$ exists, it follows that $\C, X \models K_a (\phi^\top\to\psi)$.
\end{proof}

\begin{proposition} \label{prop.standard} Let $\phi,\psi \in \lang_K$ and $a \in A$ be given. 
\begin{itemize}
%\item $\models K_a \phi \imp \M_a \phi$
\item $\models K_a \phi \imp \phi$
\item $\models K_a \phi \imp K_a K_a \phi$
\item $\models \M_a \phi \imp K_a \M_a \phi$
\end{itemize} \vspace{-.6cm}
\end{proposition}
\begin{proof} Let $\C = (C,\chi,\ell)$ and $X \in C$ be given.
\begin{itemize}

\item Let $\C,X \bowtie K_a\phi$ and $\C,X \bowtie \phi$ (recall the semantics for implication from Lemma~\ref{lemma.diamond}). Now assume $\C,X \models K_a \phi$. Then $\C,Y \bowtie \phi$ implies $\C,Y \models \phi$ for all $Y$ with $a \in \chi(X \inter Y)$. In particular, for $X=Y$ we have assumed $\C,X \bowtie \phi$, and trivially $a \in \chi(X \inter X)$. Therefore $\C,X \models \phi$ as desired.
 
\item Let $\C,X \bowtie K_a\phi$ and $\C,X \bowtie K_a K_a\phi$. It is easy to see that $\C,X \bowtie K_a\phi$ iff $\C,X \bowtie K_a K_a\phi$, so that the first suffices. But that means that the proof assumption $\C,X \models K_a \phi$ is already sufficient, because from that it follows that $\C,X\bowtie K_a\phi$. 

In order to prove $\C,X \models K_a K_a \phi$, let $Y \in C$ with $a \in \chi(X\inter Y)$ and $\C,Y \bowtie K_a \phi$ be arbitrary. We now must show that $\C,Y \models K_a \phi$. 
% From $\C,Y \bowtie K_a \phi$ follows that also $a \in \chi(Y)$. 
In order to prove that, let $Z \in C$ with $a \in \chi(Y\inter Z)$ and $\C,Z \bowtie \phi$ be arbitrary. We now must prove that $\C,Z \models \phi$. From $a \in \chi(X\inter Y)$ and $a \in \chi(Y\inter Z)$ follows $a \in \chi(X\inter Z)$. From $a \in \chi(X\inter Z)$, $\C,Z \bowtie \phi$ and assumption $\C,X\models K_a \phi$ then follows $\C,Z \models \phi$. This fulfils part of our proof requirement. 

However, in order to prove $\C,X \models K_a K_a \phi$ we still need to show that there is $V \in C$ with $a \in \chi(X\inter V)$ and $\C,V \models K_a\phi$. But this is elementary: take $V=X$.

Therefore $\C,X \models K_a K_a \phi$.
 
\item In this case it suffices to assume $\C,X \models \M_a \phi$ (as from this it follows that $\C,X \bowtie \M_a \phi$, as well as $\C,X \bowtie K_a \M_a \phi$, so that the implication is defined in $\C,X$). Then $\C,Y \models \phi$ for some $Y \in C$ with $a \in \chi(X\inter Y)$. 

In order to prove $\C,X \models K_a \M_a \phi$, first assume arbitrary $Z \in C$ with $a \in \chi(X\inter Z)$ and $\C,Z \bowtie \M_a\phi$. We need to prove that $\C,Z \models \M_a \phi$. From $a \in \chi(X \inter Y)$ and $a \in\chi(X \inter Z)$ it follows that $a \in\chi(Z \inter Y)$, and as we already obtained that $\C,Y \models \phi$, the required $\C,Z \models \M_a \phi$ follows.

However, in order to prove $\C,X \models K_a \M_a \phi$ we still need to show that there is $V \in C$ with $a \in \chi(X\inter V)$ and $\C,V \models \M_a\phi$. As in the previous case, it suffices to choose $V=X$.

\end{itemize}
\vspace{-.8cm}
\end{proof}

\begin{proposition}\label{prop.mp} Let $\phi,\psi \in \lang_K$ be given. Then $\models \phi \imp \psi$ and $\models \phi$ imply $\models\phi^\top\imp \psi$.
\end{proposition}
\begin{proof}
Let $(\C,X)$ be given such that $\C, X \bowtie \phi^\top \imp \psi$. Then, both $\C, X \bowtie \phi^\top$ and $\C, X \bowtie \psi$. The former implies, by Lemma~\ref{lem:equidef}, that $\C, X \bowtie \phi$. Thus, $\C, X \bowtie \phi \imp \psi$. It follows from the validity of $\phi\to\psi$ and $\phi$ that $\C, X \models  \phi \imp \psi$ and $\C, X \models\phi$. Thus, by Lemma~\ref{lem:MPloc}, we conclude that $\C, X \models \psi$. From that and $\C, X \bowtie \phi^\top \imp \psi$ it follows that $\C, X \models \phi^\top \imp \psi$, as required. 
\end{proof}

\begin{proposition}\label{prop.nec} Let $a \in A$ and $\phi\in\lang_K$ be given. Then $\models \phi$ implies $\models K_a \phi$.
\end{proposition}
\begin{proof}
Let $(\C,X)$ be given such that $\C,X \bowtie K_a \phi$. We have to show that $\C,X \models K_a \phi$. We have two proof obligations.

Firstly, let $Y$ be arbitrary such that $a \in \chi(X \inter Y)$ and assume that $\C,Y \bowtie \phi$. From $\C,Y \bowtie \phi$ and $\models \phi$ it follows that $\C,Y \models \phi$. Therefore, for all $Y$ with $a \in \chi(X \inter Y)$, $\C,Y \bowtie \phi$ implies $\C,Y \models \phi$.

Secondly, from $\C,X \bowtie K_a \phi$ it follows that there is a $Z$ with $a \in \chi(X \inter Z)$ such that $\C,Z \bowtie \phi$. From $\C,Z \bowtie \phi$ and $\models \phi$ it follows that $\C,Z \models \phi$. Therefore, there exists a $Z$ with $a \in \chi(X \inter Z)$ such that $\C,Z \models \phi$.

Therefore, $\C,X \models K_a \phi$.
\end{proof}

\begin{theorem}
Axiomatization {\bf S5}$^\top$ is sound.
\end{theorem}
\begin{proof}
Directly from Propositions~\ref{prop.local}, \ref{prop.k}, \ref{prop.standard}, \ref{prop.mp}, and \ref{prop.nec}.
\end{proof}

It is instructive to see why the usual {\bf K} axiom and {\bf MP} rule are invalid for our semantics. This insight might help to justify the unusual shape of {\bf K}$^\top$ and {\bf MP}$^\top$. We make the relevant counterexamples into formal propositions.

\begin{proposition}\label{prop.nonk}
There are $\phi,\psi\in \lang_K$ and $a \in A$ such that $\not\models K_a (\phi\imp\psi) \imp K_a \phi \imp K_a \psi$.
\end{proposition}
\begin{proof}
Consider again the simplicial model from Example~\ref{example.xxx}, reprinted here. 
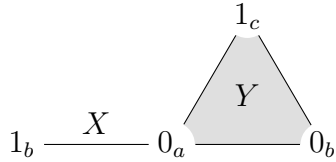
\begin{center}
\begin{tikzpicture}[round/.style={circle,fill=white,inner sep=1}]
%\fill[fill=gray!25!white] (0,0) -- (2,0) -- (1,1.71) -- cycle;
\fill[fill=gray!25!white] (2,0) -- (4,0) -- (3,1.71) -- cycle;
\node[round] (b1) at (0,0) {$1_b$};
\node[round] (b0) at (4,0) {$0_b$};
\node[round] (c1) at (3,1.71) {$1_c$};
\node[round] (a0) at (2,0) {$0_a$};
\node (f1) at (3,.65) {$Y$};
\draw[-] (b1) -- node[above] {$X$} (a0);
\draw[-] (a0) -- (b0);
\draw[-] (b0) -- (c1);
\draw[-] (a0) -- (c1);
\end{tikzpicture}
\end{center}
We can observe that:
\begin{itemize}
\item $\C,X \models K_a (p_c \imp \neg p_b)$ \\ This is because $Y \in C$ is the single facet in $C$ such that $\C,Y \bowtie p_c \imp \neg p_b$ (as this requires $\C,Y \bowtie p_c$ and $\C,Y \bowtie p_b$) and $a \in \chi(X \inter Y)$; and $\C,Y \models p_c \imp \neg p_b$ because $\C,Y \models p_c$ and $\C,Y \models \neg p_b$. 

\item $\C,X \models K_a p_c$ \\ This is for the same reason as the previous item.

\item $\C,X \not\models K_a \neg p_b$ \\ This is because there are two facets where $\neg p_b$ (and thus $p_b$) is defined: $\C,X \bowtie p_b$ and $\C,Y \bowtie p_b$. However, although $\C,Y \models \neg p_b$, $\C,X \models p_b$, and therefore $\C,X \not\models K_a \neg p_b$.
\end{itemize}
From the above we conclude that $\C,X \not\models K_a (p_c \imp \neg p_b) \imp K_a p_c \imp K_a \neg p_b$, although $\C,X \bowtie K_a (p_c \imp \neg p_b) \imp K_a p_c \imp K_a \neg p_b$. Therefore, $\not\models K_a (p_c \imp \neg p_b) \imp K_a p_c \imp K_a \neg p_b$. \end{proof}

\begin{proposition}\label{prop.nonmp}
There are $\phi,\psi\in \lang_K$ such that $\models\phi\imp\psi$ and $\models\phi$ but $\not\models\psi$.
\end{proposition}
\begin{proof}

The counterexample is
\begin{align}
\label{eq:one}
\models\quad&
 \top_a \et  \top_b \et \top_c
\\
\label{eq:two}
\models\quad&
\top_a \et \top_b \et \top_c
\imp 
\M_d ((\top_a \et p_c) \vel \M_d(\top_b \et \neg p_c))
\\
\label{eq:three}
\not\models\quad&
\M_d ((\top_a \et p_c) \vel \M_d(\top_b \et \neg p_c))
\end{align}
\paragraph*{Proof of \eqref{eq:one}.} We note that for any agent $e$ and any $(\C,X)$: $\C, X \bowtie \top_e$, iff $e \in \chi(X)$. Thus,  \eqref{eq:one} is either undefined or true. It is true iff agents $a$, $b$, and $c$ are all alive in $X$.

\paragraph*{Proof of \eqref{eq:two}.} Given some $(\C,X)$, the antecedent of the implication is defined in $X$ iff $a$, $b$, and $c$ are alive in $X$. For the consequent of the implication to be defined, it is additionally necessary to have $d$ alive in $X$. Overall, the implication \eqref{eq:two} is defined iff $\{a,b,c,d\}\subseteq \chi(X)$.

Therefore, assume $\{a,b,c,d\}\subseteq \chi(X)$. To show \eqref{eq:two} it is sufficient to show that the conclusion of the implication is true. Since $c \in \chi(X)$, $p_c$ is defined, thus, there are only two possibilities, i.e., $\C, X \models p_c$ or $\C, X \models \neg p_c$.

If $\C, X \models p_c$, then $\C, X \models \top_a \et p_c$ 
since $\top_a$ is true whenever defined. As we also have $\C, X \bowtie \M_d(\top_b \et \neg p_c)$, it follows that $\C, X \models (\top_a \et p_c)  \vel \M_d(\top_b \et \neg p_c)$. 
Thus, since $d \in \chi(X \cap X)$, we have
\[
\C, X \models \M_d ((\top_a \et p_c)  \vel \M_d(\top_b \et \neg p_c))
\]
and therefore \eqref{eq:two}.

If $\C, X \models \neg p_c$, then $\C, X \models \top_b \et \neg p_c$, 
and since $d \in \chi(X \cap X)$, also $
\C, X \models \M_d(\top_b \et \neg p_c)$. 
Given $\C, X \bowtie \top_a \et p_c$, it again follows that $\C, X \models (\top_a \et p_c)  \vel \M_d(\top_b \et \neg p_c)$ and the argument for \eqref{eq:two} being  true is now the same as in the preceding case.

Thus, we showed that in either case \eqref{eq:two} is  true.
\paragraph*{Proof of \eqref{eq:three}.} 
Consider the simplicial model $\C$ depicted below, wherein we have only labelled the $c$ nodes with the value of the agent's propositional variable (we do not care about the value of the variables labelling the other nodes, only about the colour of the nodes).
\begin{figure}[ht]
\center
\scalebox{1}{
\begin{tikzpicture}[round/.style={circle,fill=white,inner sep=1}]
%\fill[fill=gray!25!white] (0,0) -- (2,0) -- (1,1.71) -- cycle;
\fill[fill=gray!25!white] (2,0) -- (4,0) -- (3,1.71) -- cycle;
\fill[fill=gray!25!white] (0,0) -- (2,0) -- (1,1.71) -- cycle;
\node[round] (b1) at (0,0) {$b$};
\node[round] (b0) at (4,0) {$a$};
\node[round] (c1) at (3,1.71) {$0_c$};
\node[round] (lc1) at (1,1.71) {$1_c$};
\node[round] (a0) at (2,0) {$d$};
\node (f1) at (3,.65) {$Y$};
\node (f1) at (1,.65) {$X$};
\node(c) at (-1,0) {$\C:$};
\draw[-] (b1) -- (a0);
\draw[-] (b1) -- (lc1);
\draw[-] (a0) -- (lc1);
\draw[-] (a0) -- (b0);
\draw[-] (b0) -- (c1);
\draw[-] (a0) -- (c1);
\end{tikzpicture}
}
\end{figure}

For this complex we can make the following observations in tandem, schematically listed in the order of justification, thus producing a `witness' against the validity of \eqref{eq:three} (and even two, merely for the sake of symmetry).
\begin{align*}
\C, X &\bowtie \top_b \et \neg p_c 
&
\C, Y &\not\bowtie \top_b \et \neg p_c 
\\
\C, X &\not\models \top_b \et \neg p_c 
&
%\C, Y &\not\bowtie \top_b \et \neg p_c 
\\
 &
&
\C, Y &\bowtie \M_d(\top_b \et \neg p_c) 
\\
 &
&
\C, Y &\not\models \M_d(\top_b \et \neg p_c) 
\\
\C, X &\not\bowtie \top_a \et  p_c 
&
\C, Y &\bowtie \top_a \et p_c 
\\
 &
&
\C, Y &\not\models \top_a \et p_c 
%\end{align*}
%\begin{align*}
\\
\C, X &\not\bowtie (\top_a \et p_c) \vel  \M_d(\top_b \et \neg p_c)
&
\C, Y &\bowtie (\top_a \et p_c) \vel  \M_d(\top_b \et \neg p_c)
\\
&
&
\C, Y &\not\models (\top_a \et p_c) \vel  \M_d(\top_b \et \neg p_c)
\\
\C, X &\bowtie \M_d ((\top_a \et p_c) \vel  \M_d(\top_b \et \neg p_c))
&
\C, Y &\bowtie \M_d ((\top_a \et p_c) \vel  \M_d(\top_b \et \neg p_c))
\\
\C, X &\not\models \M_d ((\top_a \et p_c) \vel  \M_d(\top_b \et \neg p_c))
&
\C, Y &\not\models \M_d ((\top_a \et p_c) \vel  \M_d(\top_b \et \neg p_c))
\end{align*}
We have shown that \eqref{eq:three} is defined but false and is therefore not valid.
\end{proof}

\section{Pure complexes}  \label{section.pure}

When the simplicial complexes are pure and of dimension $|A|-1=n$, the logical semantics should become that of \cite{goubaultetal:2018,ledent:2019} and the logic should become {\bf S5}{+}{\bf L}. We show that this is indeed the case.

\begin{lemma} \label{lemma.pure}
Let $\C = (C,\chi,\ell)$ be pure and $\phi\in\lang_K$.
\begin{itemize}
\item For all $Y \in \FF(C)$: $\C,Y\bowtie\phi$.
\item For all $X \in C$ with $a \in \chi(X)$: $\C,X\bowtie\M_a\phi$.
\end{itemize} \vspace{-.6cm}
\end{lemma}
\begin{proof}
The first item is shown by induction on $\phi$. If $\phi = p_a$ for some $a \in A$ and $p_a \in P_a$, then, as $C$ is pure and its facets $Y$ contain all colours, $a \in \chi(Y)$ and therefore $\C,Y \bowtie p_a$. The cases conjunction and negation are trivial. If $\phi = \M_a\phi'$ for some $a \in A$, then $\C,Y\bowtie \M_a\phi'$ iff there is $X$ with $a \in \chi(X\inter Y)$ and $\C,X\bowtie\phi'$. We now take $X=Y$ and we may assume $\C,Y\bowtie\phi'$ by induction.

The second is shown using the first item. By definition, $\C,X\bowtie\M_a\phi$ iff there is $Z$ with $a \in \chi(X\inter Z)$ and $\C,Z\bowtie\phi$. Take any $Z \in \FF(C)$ with $X \subseteq Z$. By the first item it holds that $\C,Z\bowtie\phi$. Therefore, $\C,X\bowtie\M_a\phi$. (For a vertex $v$ coloured with $a$ we therefore now always have that $\C,v\bowtie\M_a\phi$. However, if $\chi(v)\neq a$ then also for pure complexes we still have that $\C,v\not\bowtie\M_a\phi$!)
\end{proof}

Instead of a satisfaction relation $\models$ between pairs $(\C,X)$ consisting of a (possibly impure) simplicial model $\C$ for agents $A$ and a simplex $X$ contained in it, and a formula $\phi \in \lang_K(A)$, we now define a satisfaction relation $\models_\pure$ between pairs $(\C,X)$ consisting of a pure simplicial model $\C$ for agents $A$ of dimension $|A|-1=n$ and a facet $X$ contained in it, and a formula $\phi \in \lang_K(A)$. Note that the logical language is the same either way. We will show that this satisfaction relation $\models_\pure$ is as in \cite{goubaultetal:2018}.

First we recall the definition of $\models$ (Definition~\ref{def.defsat}). As \cite{goubaultetal:2018} has $K_a\phi$ as linguistic primitive, not $\M_a\phi$, we will do likewise.
\[ \begin{array}{lcl}
\C, X \models p_a & \text{iff} & a \in \chi(X) \ \text{and} \ p_a \in \ell(X) \\
\C, X \models \phi\et\psi & \text{iff} & \C, X \models \phi \ \text{and} \ \C, X \models \psi \\
\C, X \models \neg \phi & \text{iff} & \C, X \bowtie \phi \ \text{and} \ \C, X \not\models \phi \\
\C,X \models K_a\phi & \text{iff} & \C,X \bowtie K_a\phi \ \text{and} \\ && \C,Y \bowtie \phi \text{ implies } \C,Y \models \phi \ \text{for all} \ Y \in C \ \text{with} \ a \in \chi(X \inter Y)
\end{array}\]
When $X$ and $Y$ above are facets in a pure complex $\C$ of dimension $n$, in view of Lemma~\ref{lemma.pure} we can scrap all definability requirements and we thus obtain
\[ \begin{array}{lcl}
\C, X \models_\pure p_a & \text{iff} & p_a \in \ell(X) \\
\C, X \models_\pure \phi\et\psi & \text{iff} & \C, X \models_\pure \phi \ \text{and} \ \C, X \models_\pure \psi \\
\C, X \models_\pure \neg \phi & \text{iff} & \C, X \not\models_\pure \phi \\
\C,X \models_\pure K_a\phi & \text{iff} & \C,Y \models_\pure \phi \ \text{for all} \ Y \in \FF(C) \ \text{with} \ a \in \chi(X \inter Y)
\end{array}\]
which is the semantics of \cite{goubaultetal:2018}.

Even on pure complexes, there is nothing against privileging the $\models$ semantics, because it is local and allows us to interpret formulas in, for example, vertices, which seems a natural thing to do. We then continue to need definability requirements such as $\C, v \not\bowtie p_a$ if $a \neq \chi(v)$.

All prior results obtained for impure complexes are obviously preserved for pure complexes. In particular we note that semantic equivalence of formulas $\phi,\psi$ and validity of formulas $\phi$ becomes as usual, when restricting the definitions to pure complexes of dimension $n$ and facets. This was (see again Definition~\ref{def.defsat})
\begin{quote}
Given $\phi,\psi\in\lang_K$, $\phi$ is {\em equivalent} to $\psi$ if for all $(\C,X)$: $\C,X \models \phi$ iff $\C,X \models \psi$, $\C,X \models \neg\phi$ iff $\C,X \models \neg\psi$, and $\C,X \not\bowtie\phi$ iff $\C,X\not\bowtie\psi$. A formula $\phi\in\lang_K$ is {\em valid} if for all $(\C,X)$: $\C,X \bowtie \phi$ implies $\C,X \models \phi$.
\end{quote}
and now has become, for facets $X$ only,
\begin{quote}
Given $\phi,\psi\in\lang_K$, $\phi$ is {\em equivalent} to $\psi$ if for all $(\C,X)$: $\C,X \models_\pure \phi$ iff $\C,X \models_\pure \psi$. A formula $\phi\in\lang_K$ is {\em valid} if for all $(\C,X)$: $\C,X \models_\pure \phi$.
\end{quote}
% where we note that any $(\C,X)$ as below where $X$ is a facet is also a $(\C,X)$ as above where $X$ may be any simplex (we have class inclusion). 
This should not be surprising, as the $\models_\pure$ semantics is exactly the \cite{goubaultetal:2018} semantics.
 
This brings us to the axiomatization. Again, as the $\models_\pure$ semantics is the \cite{goubaultetal:2018} semantics, the logic must be the same, that is, {\bf S5} plus the locality axiom {\bf L}. The difference between {\bf S5} and {\bf S5}$^\top$ only concerns the {\bf K} axiom and the {\bf MP} derivation rule. We recall that {\bf S5}$^\top$ contained
\[\begin{array}{ll}
\mathbf{K}^\top \qquad & \models K_a (\phi\imp\psi) \imp K_a \phi \imp K_a (\phi^\top \imp \psi) \\
\mathbf{MP}^\top \qquad  & \text{From } \models \phi\imp\psi \text{ and } \models \phi, \text{infer } \models \phi^\top \imp \psi
\end{array}\]
These are now replaced by
\[\begin{array}{ll}
\mathbf{K} \qquad & \models_\pure K_a (\phi\imp\psi) \imp K_a \phi \imp K_a \psi \\
\mathbf{MP} \qquad  & \text{From } \models_\pure \phi\imp\psi \text{ and } \models_\pure \phi, \text{infer } \models_\pure \psi
\end{array}\]

So we get the complete axiomatization $\mathbf{S5}+\mathbf{L}$ for free, by referring to \cite{goubaultetal:2018}. 

We should note that is not surprising\dots Validity of all axioms in {\bf S5}$^\top$ is preserved, as truth for all $(\C,X)$ where $\C$ may be impure and $X$ is any simplex, implies truth for all $(\C',X')$ where $\C'$ is pure of dimension $n$ and $X$ a facet. Also {\bf K}$^\top$ remains valid, and clearly $\mathbf{K}^\top \eq \mathbf{K}$ is $\models_\pure$ valid.  Validity preservation of the derivation rules in {\bf S5}$^\top$ also holds but was not guaranteed, as the assumption of truth for all $(\C,X)$ where $\C$ may be impure and $X$ is any simplex is stronger than the assumption of truth for all $(\C',X')$ where $\C'$ is pure of dimension $|A|-1$ and $X$ a facet. So necessation {\bf Nec} and {\bf MP} would have to be shown anew. Here we can be lazy and for {\bf Nec} refer to \cite{goubaultetal:2018} (although an actual proof would not take more than a few lines). Concerning {\bf MP}, it suffices to observe that $\models_\pure \psi$ is equivalent to $\models_\pure \phi^\top \imp \psi$, as $\phi^\top$ (that is now always defined) is a $\models_\pure$ validity. 

This ends our exploration into pure complexes of dimension $|A|-1$ and this is the `sanity check' result that was expected.

\section{Correspondence to Kripke models}  \label{section.correspondence}

A precise one-to-one correspondence between pure simplicial models and multi-agent Kripke models satisfying the following three conditions is given in \cite{goubaultetal:2018,ledent:2019}: $(i)$ all accessibility relations are equivalences, $(ii)$ all propositional variables are local, that is, there is an agent who knows the value of that variable, and  $(iii)$ the intersection of all relations is the identity. We now propose Kripke models where agents may be dead {\bf or} alive. We call them {\em local epistemic models}.

For local epistemic models we require that for each agent there is a subset of the domain on which the accessibility relation for that agent is an equivalence relation (these are the states where that agent is alive), whereas in the remainder of the domain the accessibility relation is empty (these are the states where that agent is dead). We also require that for any two distinct states there must always be a live agent in one that can distinguish it from the other and a live (possibly but not necessarily different) agent in the other that can distinguish it from the one, generalizing the requirement for pure complexes that the intersection of relations is the identity \cite{goubaultetal:2018}. (This requirement should be credited to \cite{goubaultetal:2021}.) Additionally we require that a formula can only be interpreted in a given state  of a model if the interpretation is defined in that state. For example, we do not wish to say that formula $p_a$ is true or false in a state where agent $a$'s accessibility relation is empty, $p_a$ should be undefined in that state. 

Given that, however, and with the epistemic logic for impure complexes already at our disposal, we can map simplicial models to local epistemic models and vice versa, and these transformations are truth preserving. 

\paragraph*{Local epistemic models.}

As before, given are the set $A$ of agents and the set $P = \Union_{a \in A} P_a$ of (local) variables. Given an abstract domain of objects called {\em states} and an agent $a$, a {\em local equivalence relation} ($\sim\!(a)$ or) $\sim_a$ is a binary relation between elements of $S$ that is an equivalence relation on a subset of $S$ denoted $S_a$ and otherwise empty. So, $\sim_a$ induces a partition on $S_a$, whereas ${\sim_a} = \emptyset$ on the complement $\overline{S}_a := S \setminus S_a$ of $S_a$. For $(s,t) \in {\sim_a}$ we write $s \sim_a t$, and for $\{ t \mid s \sim_a t \}$ we write $[s]_a$: this is an equivalence class of the relation $\sim_a$ on $S_a$. Given $s \in S$, let $A_s := \{a \in A \mid s \in S_a\}$. Set $A_s$ contains the agents that are alive in state $s$. Note that $a \in A_s$ iff $s \in S_a$.

\begin{definition}[Local epistemic model]
\emph{Local epistemic frames} are pairs $\model = (S,\sim)$ where $S$ is the (non-empty) domain of \emph{(global) states}, and $\sim$ is a function that maps the agents $a\in A$ to local equivalence relations $\sim_a$ that are required to be \emph{proper}, that is: for all distinct $s,t \in S$ there is a $b \in A_s$ such that $s \not\sim_b t$ and there is (therefore also) a $c \in A_t$ such that $s \not\sim_c t$. Agents $b$ and $c$ may but need not be the same. \emph{Local epistemic models} are triples $\model = (S,\sim,L)$, where $(S,\sim)$ is a local epistemic frame, and where \emph{valuation} $L$ is a function from $S$ to $\powerset(P)$ satisfying 
%the following requirements: \begin{itemize}
%\item For all $a \in A$, $p_a \in P_a$ and $s \in \overline{S}_a$, $p_a \notin L(s)$.
%\item 
that for all $a \in A$, $p_a \in P_a$ and $s,t \in S_a$, if $s \sim_a t$ then $p_a \in L(s)$ iff $p_a \in L(t)$. We say that all variables $p_a$ are {\em local} for agent $a$. 
A pair $(\model,s)$ where $s \in S$ is a \emph{pointed} local epistemic model.
%\end{itemize}
\end{definition}
Note that there is no requirement for the valuation of variables $p_a$ on the complement $\overline{S}_a$. 
% It should be noted that the requirement that the valuation of local variables $p_a$ is empty on $\overline{S}_a$ could be dropped, given the logical semantics to follow it is rather that the valuation on that part of the domain does not matter for those variables. The chosen restriction seems a cleaner modelling solution as it does not beg that question.

% An epistemic frame $(S,\sim)$ is \emph{proper} if ${\inter_{a \in A} \sim_a} = Id$, where the identity relation $Id : = \{(s,s) \mid s \in S\}$. An epistemic model is proper if its underlying frame is proper. 

\paragraph*{Semantics on local epistemic models.}

The interpretation of a formula $\phi\in \lang_K$ in a global state of a given pointed local epistemic model $(\model,s)$ is by induction on the structure of $\phi$. As before, we need relations $\bowtie$ to determine whether the interpretation is defined, and $\models$ to determine its truth value when defined. % The expression ``$\model,s\models \phi$'' stands for ``in global state $s$ of epistemic model $\model$ it holds that (or: it is true that) $\phi$.''
\begin{definition} Let $\model = (S,\sim,L)$ be given. We define $\bowtie$ and $\models$ by induction on $\phi\in\lang_K$.
\[ \begin{array}{lcl}
\model,s \bowtie p_a & \text{iff} & s \in S_a \\
\model,s \bowtie \neg\phi & \text{iff} & \model,s \bowtie \phi \\
\model,s \bowtie \phi\et\psi & \text{iff} & \model,s \bowtie \phi \text{ and } \model,s \bowtie \psi \\
\model,s \bowtie \M_a \phi & \text{iff} & \model,t \bowtie \phi \text{ for some } t \text{ with } s \sim_a t
\end{array} \]

\[ \begin{array}{lcl}
\model,s \models p_a & \text{iff} & s \in S_a \text{ and } p_a \in L(s) \\
\model,s \models \neg\phi & \text{iff} & \model,s \bowtie \phi \text{ and } \model,s \not\models \phi \\
\model,s \models \phi\et\psi & \text{iff} & \model,s \models \phi \text{ and } \model,s \models \psi \\
\model,s \models \M_a \phi & \text{iff} & \model,t \models \phi \text{ for some } t \text{ with } s \sim_a t
\end{array} \]
Formula $\phi$ is {\em valid} iff for all $(\model,s)$, $\model,s \bowtie \phi$ implies $\model,s \models \phi$. We let $\I{\phi}_\model$ stand for $\{ s \in S \mid \model,s \models \phi \}$. This set is called the {\em denotation} of $\phi$ in $\model$.
\end{definition}

Analogous results as for the semantics on simplicial complexes can be obtained for the semantics on local epistemic models, demonstrating the tricky interaction between $\bowtie$ and $\models$. For example, the interpretation of other propositional connectives such as disjunction and that of knowledge, now becomes (we recall Lemma~\ref{lemma.diamond}):
%\[ \begin{array}{llll}
%\text{---} & \model,s \models \phi \vel \psi & \text{iff} & \model,s \bowtie \phi, \model,s \bowtie \psi, \ \text{and } \model,s \models \phi \ \text{or} \ \model,s \models \psi \\
%\text{---} & \model,s \models K_a\phi & \text{iff} & \model,s \bowtie K_a\phi, \ \text{and} \\ &&& \model,t \bowtie \phi \text{ implies } \model,t \models \phi \ \text{for all} \ t \in S \text{with} \ s \sim_a t
%\end{array}\]

\[ \begin{array}{llll}
\model,s \models \phi \vel \psi & \text{iff} & \model,s \bowtie \phi, \model,s \bowtie \psi, \ \text{and } \model,s \models \phi \ \text{or} \ \model,s \models \psi \\
\model,s \models K_a\phi & \text{iff} & \model,s \bowtie K_a\phi, \ \text{and} \\ && \model,t \bowtie \phi \text{ implies } \model,t \models \phi \ \text{for all} \ t \in S \ \text{with} \ s \sim_a t
\end{array}\]
Instead of showing all this in detail, we restrict ourselves to show a truth value (and definability) preserving transformation between simplicial models and local epistemic models.

\paragraph*{Correspondence between simplicial models and epistemic models.}

Given agents $A$ and variables $P$, let $\mathcal K$ be the class of local epistemic models and let $\mathcal S$ be the class of (possibly impure) simplicial models. In \cite{goubaultetal:2018}, the \emph{pure} simplicial models are shown to correspond to local epistemic models of the following special kind: all relations are equivalence relations on the entire domain of the model (instead of on a subset only), and the intersection of the relations for all agents is the identity, a frame requirement they call \emph{proper} (and that is not bisimulation invariant). We give a generalization of their construction where the relations merely need to be equivalence relations on a subset of the domain.

We define maps $\sigma: \mathcal K \imp \mathcal S$ ($\sigma$ for \emph{S}implicial) and $\kappa: \mathcal S \imp \mathcal K$ ($\kappa$ for \emph{K}ripke), such that $\sigma$ maps each local epistemic model $\model$ to a simplicial model $\sigma(\model)$, and $\kappa$ maps each  simplicial model $\C$ to a local epistemic model $\kappa(\C)$. As $\sigma$ maps a state $s$ in $\model$ to a facet $X=\sigma(s)$ in $\sigma(\model)$, and $\kappa$ maps each facet $X$ in $\C$ to a state $s = \kappa(X)$ in $\kappa(\C)$, these maps are also between pointed structures $(\model,s)$ respectively $(\C,X)$.  Subsequently we then show that for all $\phi\in\lang_K$, $\model,s\models\phi$ iff $\sigma(\model,s) \models\phi$, and that (for facets $X$) $\C,X\models\phi$ iff $\kappa(\C,X)\models \phi$.

\begin{definition}
Given a local epistemic model $\model = (S,\sim,L)$, we define $\sigma(\model) = (C,\chi,\ell)$: % as follows. 
\[\begin{array}{lll}
X \in C &\text{iff}& X = \{ ([s]_a,a) \mid a \in B \} \text{ for some } s \in S \text{ and } B \subseteq A \text{ with } \emptyset\neq B \subseteq A_s \\
\chi(([s]_a,a)) &=& a \\
p_a \in \ell(([s]_a,a)) & \text{iff} & p_a \in L(s)
\end{array}\vspace{-.6cm}\]
\end{definition}
It follows from the definition of $\sigma(\model)$ that:
\[\begin{array}{lll}
\VV(C) &=& \{ ([s]_a,a) \mid s \in A, a \in A_s\} \\
X \in \FF(C) &\text{iff}& X = \{ ([s]_a,a) \mid a \in A_s\} \text{ for some } s \in S \text{ with } A_s \neq \emptyset
\end{array}\]
Note that in a vertex denoted $([s]_a,a)$ the first argument $[s]_a$ is a set of states in which the name $a$ of the agent does not appear, which is why we need the second argument $a$ to determine its colour in $\chi(([s]_a,a))=a$.

Given $s \in S$, we let $\sigma(s)$ denote the facet $\{ ([s]_a,a) \mid a \in A_s\}$, and by $\sigma(\model,s)$ we mean $(\sigma(\model),\sigma(s))$. The requirement that local epistemic models are proper ensures that different states are mapped to different facets, that is, for $s \neq t$ we have that neither $\sigma(s)\subseteq\sigma(t)$ nor $\sigma(t)\subseteq\sigma(s)$. This is a generalization of a similar requirement in \cite{goubaultetal:2018,ledent:2019} for epistemic models corresponding to pure complexes, namely that $\Inter_{a \in A}\sim_a$ is the identity relation. Some issues with this requirement are discussed at the end of this section.

Finally we should observe that $\sigma(\model)$ is indeed a simplicial model. This is elementary to see. It is closed under subsets of simplices. Also, $([s]_a,a) = ([t]_a,a)$ means that $s \sim_a t$, and the locality of the epistemic model then guarantees that $p_a \in L(s)$ iff $p_a \in L(t)$.

\begin{definition}
Given a simplicial model $\C = (C,\chi,\ell)$, we define $\kappa(\C) = (S,\sim,L)$: %, where $X,Y \in \FF(C)$: % as follows, where $a \in A$, $p_a \in P_a$, and $X,Y \in \FF(C)$. 
\[\begin{array}{lll}
S & = & \FF(C) \\
X \sim_a Y & \text{iff} & a \in \chi(X \inter Y) \\
p_a \in L(X) & \text{iff} & p_a \in \ell(X)
\end{array}\]
\end{definition}
Simplices $X,Y$ above are elements of the domain, and therefore facets. We let $\kappa(\C,X)$ for facets $X$ denote $(\kappa(\C),X)$. Concerning the definition of $\kappa(\C)$ we recall that $\ell(X) = \Union_{a \in \chi(X)} \ell(X_a)$, where for a vertex $v$ coloured $a$ such as $X_a$, $\ell(v)\subseteq P_a$. Given $a \in A$, a facet not containing that colour (so of dimension less than $|A|-1$), that is, $X \in \FF(C)$ with $a \notin \chi(X)$, is in the $\overline{S}_a$ part of the model $\kappa(\C)$, on which ${\sim_a} = \emptyset$. Whereas restricted to $S_a$, $\sim_a$ is indeed an equivalence relation between facets/states. The states in $S_a$ may also be facets of dimension less than $|A|-1$, but then they lack a vertex for a colour other than $a$.

If $a \notin \chi(X)$, then obviously $p_a \notin \ell(X)$. Consequently, in $\kappa(\C)$ the valuation of variables for agent $a$ in the $\overline{S}_a$ part of the model is empty. This was not required, but it does not matter, as we will show that it is irrelevant for the truth (or definability) of formulas.

Finally, note that $\kappa(\C)$ is indeed a local epistemic model. Here it is important to observe that the model is proper: distinct facets $X,Y \in C$ intersecting in $a$ are mapped to different states in $\kappa(\C)$. As $X\setminus Y \neq \emptyset$ there is a $b \in \chi(X\setminus Y)$ and therefore $X \not\sim_b Y$ in $\kappa(\C)$. As $Y\setminus X \neq \emptyset$ there is a $c \in \chi(Y\setminus X)$ and therefore $Y \not\sim_c X$ in $\kappa(\C)$.

\begin{figure}
\center
\scalebox{.8}{
\begin{tabular}{cccccc}
%&&
\begin{tikzpicture}[round/.style={circle,fill=white,inner sep=1}]
%\fill[fill=gray!25!white] (0,0) -- (2,0) -- (1,1.71) -- cycle;
\fill[fill=gray!25!white] (2,0) -- (4,0) -- (3,1.71) -- cycle;
\node[round] (b1) at (0,0) {$1_b$};
\node[round] (b0) at (4,0) {$0_b$};
\node[round] (c1) at (3,1.71) {$1_c$};
\node[round] (a0) at (2,0) {$0_a$};
%\node (f1) at (3,.65) {$Y$};
%
\draw[-] (b1) --  (a0);
\draw[-] (a0) -- (b0);
\draw[-] (b0) -- (c1);
\draw[-] (a0) -- (c1);
\end{tikzpicture} 
& \quad $\stackrel \kappa \Imp$ \quad &
\begin{tikzpicture}
\node (010l) at (0,0.4) {\scriptsize$ab$};
\node (001l) at (3,0.4) {\scriptsize$abc$};
\node (010) at (.5,0) {$0_a1_b0_c$};
\node (001) at (3.5,0) {$0_a0_b1_c$};
\draw[-] (010) -- node[above] {$a$} (001);
\end{tikzpicture}
\\
%&&
\begin{tikzpicture}[round/.style={circle,fill=white,inner sep=1}]
%\fill[fill=gray!25!white] (0,0) -- (2,0) -- (1,1.71) -- cycle;
\fill[fill=gray!25!white] (2,0) -- (4,0) -- (3,1.71) -- cycle;
\fill[fill=gray!25!white] (0,0) -- (2,0) -- (1,1.71) -- cycle;
\node[round] (b1) at (0,0) {$1_b$};
\node[round] (b0) at (4,0) {$0_b$};
\node[round] (c1) at (3,1.71) {$1_c$};
\node[round] (lc1) at (1,1.71) {$0_c$};
\node[round] (a0) at (2,0) {$0_a$};
%\node (f1) at (3,.65) {$Y'$};
%\node (f1) at (1,.65) {$W'$};
%
\draw[-] (b1) -- (a0);
\draw[-] (b1) -- (lc1);
\draw[-] (a0) -- (lc1);
\draw[-] (a0) -- (b0);
\draw[-] (b0) -- (c1);
\draw[-] (a0) -- (c1);
\end{tikzpicture}
& \quad $\stackrel \kappa \Imp$ \quad &
\begin{tikzpicture}
\node (010l) at (0,0.4) {\scriptsize$abc$};
\node (001l) at (3,0.4) {\scriptsize$abc$};
\node (010) at (.5,0) {$0_a1_b0_c$};
\node (001) at (3.5,0) {$0_a0_b1_c$};
\draw[-] (010) -- node[above] {$a$} (001);
\end{tikzpicture}
\\
%&&
\begin{tikzpicture}[round/.style={circle,fill=white,inner sep=1}]
%\fill[fill=gray!25!white] (0,0) -- (2,0) -- (1,1.71) -- cycle;
%\fill[fill=gray!25!white] (2,0) -- (4,0) -- (3,1.71) -- cycle;
\node[round] (b1) at (0,0) {$1_b$};
\node[round] (b0) at (4,0) {\color{white}$0_b$};
\node[round] (c1) at (3,1.71) {$1_c$};
\node[round] (a0) at (2,0) {$0_a$};
%\node (f1) at (3,.65) {$Y$};
%
\draw[-] (b1) --  (a0);
%\draw[-] (a0) -- node[above] {$Z''$} (b0);
%\draw[-] (b0) -- (c1);
\draw[-] (a0) --  (c1);
\end{tikzpicture}
& \quad $\stackrel \kappa \Imp$ \quad  &
\begin{tikzpicture}
\node (010l) at (0,0.4) {\scriptsize$ab$};
\node (001l) at (3,0.4) {\scriptsize$ac$};
\node (010) at (.5,0) {$0_a1_b0_c$};
\node (001) at (3.5,0) {$0_a0_b1_c$};
\draw[-] (010) -- node[above] {$a$} (001);
\end{tikzpicture}
& \quad $\stackrel \sigma \Imp$ \qquad &
\begin{tikzpicture}[round/.style={circle,fill=white,inner sep=1}]
\node[round] (b1) at (0,0) {$1_b$};
\node[round] (c1) at (3,1.71) {$1_c$};
\node[round] (a0) at (2,0) {$0_a$};
\node[round] (b0) at (4,0) {\color{white}$0_b$};
\node (b1b) at (-.5,.4) {\scriptsize$(\{0_a1_b0_c\},b)$};
\node (c1b) at (1.9,1.51) {\scriptsize$(\{0_a0_b1_c\},c)$};
\node (a0b) at (3.6,.2) {\scriptsize$(\{0_a1_b0_c,0_a0_b1_c\},a)$};
\draw[-] (b1) --  (a0);
%\draw[-] (a0) --  (b0);
%\draw[-] (b0) -- (c1);
\draw[-] (a0) --   (c1);
\end{tikzpicture}
\\
%& & 
\begin{tikzpicture}[round/.style={circle,fill=white,inner sep=1}]
%\fill[fill=gray!25!white] (0,0) -- (2,0) -- (1,1.71) -- cycle;
\fill[fill=gray!25!white] (2,0) -- (4,0) -- (3,1.71) -- cycle;
\node[round] (b0) at (4,0) {$1_b$};
\node[round] (c1) at (3,1.71) {$0_c$};
\node[round] (a0) at (2,0) {$0_a$};
%\node (f1) at (3,.65) {$Y$};
%
\draw[-] (a0) -- (b0);
\draw[-] (b0) -- (c1);
\draw[-] (a0) -- (c1);
\end{tikzpicture}
& \quad $\stackrel \kappa \Imp$ \quad &
\begin{tikzpicture}
\node (010l) at (0,0.4) {\scriptsize$abc$};
\node (010) at (.5,0) {$0_a1_b0_c$};
\end{tikzpicture}
\end{tabular}
}
\caption{From simplicial models to local epistemic models, and vice versa in one case. Labels of states in epistemic models list the agents that are alive.}
\label{figure.kappasigma}
\end{figure}
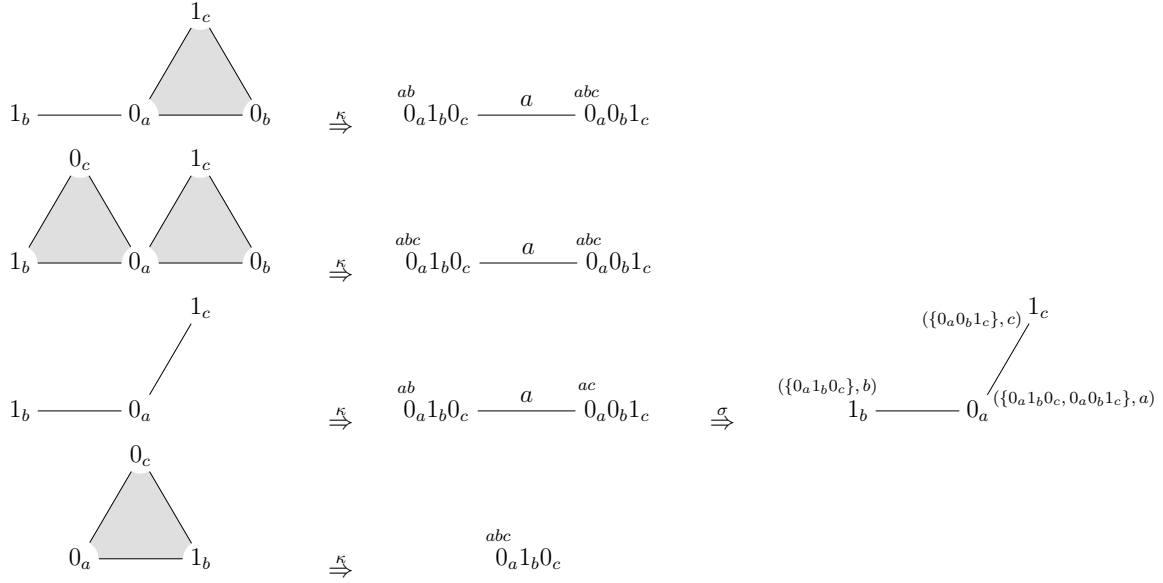
%
%\end{example}

\begin{example} \label{example.kappasigma}
Figure~\ref{figure.kappasigma} shows various examples of the transformation via $\kappa$ of simplicial models into local epistemic models. The third simplicial model consisting of two edges produces a local epistemic model that is indeed proper: agent $b$ is alive in the left state $0_a1_b0_c$ and can distinguish this state from the right state $0_a0_b1_c$, whereas agent $c$ is alive in the right state $0_a0_b1_c$ and can distinguish that state from the left state $0_a1_b0_c$. 

To obtain simplicial models from local epistemic models, we simply follow the arrow in the other direction, where the only difference is that the names of vertices are now pairs consisting of an equivalence class of states and an agent. This is demonstrated for the third model only.
\end{example}
It is easy to see that (always) $\sigma(\kappa(\C))$ is isomorphic to $\C$ and $\kappa(\sigma(\model))$ is isomorphic to $\model$.
%Without the agents' labels of states, the visualization would be ambiguous, as all two-state local epistemic models would then be identical.

\begin{proposition}\label{prop.corr2}
For all formulas $\phi \in \lang_K$, for all pointed local epistemic models $(\model,s)$: $\model,s\bowtie \phi$ iff $\sigma(\model,s) \bowtie \phi$, and $\model,s\models \phi$ iff $\sigma(\model,s) \models \phi$.
\end{proposition}
\begin{proof}
The proof is by induction on $\phi$. In the case negation for the $\models$ statement it is important that the induction hypothesis holds for the $\bowtie$ statement and for the $\models$ statement, which requires us to show them simultaneously.

\medskip
\noindent $\model,s \bowtie p_a \quad \text{ iff } \\
s \in S_a \quad \text{ iff } \text{(recall that } \sigma(s) = \{ ([s]_a,a) \mid a \in A_s\}) \\
{([s]_a,a)} \in \sigma(s) \quad \text{ iff } \\
a \in \chi(\sigma(s)) \quad \text{ iff } \\
\sigma(\model), \sigma(s) \bowtie p_a$

\medskip
\noindent $\model,s \bowtie \neg\phi \quad \text{ iff } \\
\model,s \bowtie \phi  \quad \text{ iff (by induction)} \\
\sigma(\model,s) \bowtie \phi \quad \text{ iff } \\
\sigma(\model,s) \bowtie \neg \phi$

\medskip
\noindent $\model,s \bowtie \phi\et\psi \quad \text{ iff } \\
\model,s \bowtie \phi \text{ and } \model,s \bowtie \psi \quad \text{ iff (by induction) } \\
\sigma(\model,s) \bowtie \phi \text{ and } \sigma(\model,s) \bowtie \psi \quad \text{ iff } \\
\sigma(\model,s) \bowtie \phi\et\psi$

\medskip
\noindent $\model,s \bowtie \M_a\phi \quad \text{ iff } \\
\model,t \bowtie \phi \text{ for some } t \sim_a s \quad \text{ iff (by induction)} \\
\sigma(\model),\sigma(t) \bowtie \phi \text{ for some } t \sim_a s \quad \text{ iff } \\
\sigma(\model),\sigma(t) \bowtie \phi \text{ for some } \sigma(t) \text{ with } a\in \chi(\sigma(s)\inter\sigma(t)) \quad \text{ iff (use Lemma~\ref{lemma.upbowtie} for $\Pmi$)} \\
\sigma(\model,s) \bowtie \M_a\phi$

\medskip

We continue with the satisfaction relation.

\medskip

\noindent $\model,s \models p_a \quad \text{ iff } \\
s \in S_a \text{ and } p_a \in L(s) \quad \text{ iff (as $([s]_a,a) \in \sigma(s)$, and $p_a 
\in \ell(([s]_a,a)) \subseteq \ell(\sigma(s))$)} \\
a \in \chi(\sigma(s)) \text{ and } p_a \in \ell(\sigma(s)) \quad \text{ iff } \\
\sigma(\model), \sigma(s) \models p_a$

\medskip
\noindent $\model,s \models \neg\phi \quad \text{ iff } \\
\model,s\bowtie\phi \text{ and } \model,s \not\models \phi \quad \text{ iff (by induction for $\bowtie$ and for $\models$)} \\
\sigma(\model,s) \bowtie \phi \text{ and } \sigma(\model,s) \not\models \phi \quad \text{ iff } \\
\sigma(\model,s) \models \neg \phi$

\medskip
\noindent $\model,s \models \phi\et\psi \quad \text{ iff } \\
\model,s \models \phi \text{ and } \model,s \models \psi \quad \text{ iff (by induction) } \\
\sigma(\model,s) \models \phi \text{ and } \sigma(\model,s) \models \psi \quad \text{ iff } \\
\sigma(\model,s) \models \phi\et\psi$

\medskip
\noindent $\model,s \models \M_a\phi \quad \text{ iff } \\
\model,t \models \phi \text{ for some } t \sim_a s \quad \text{ iff (by induction)} \\
\sigma(\model),\sigma(t) \models \phi \text{ for some } \sigma(t) \text{ with } a\in \chi(\sigma(s)\inter\sigma(t)) \quad \text{ iff (use Proposition~\ref{prop.star} for $\Pmi$)} \\
\sigma(\model),\sigma(s) \models \M_a\phi$
\end{proof}

\begin{proposition}\label{prop.corr3}
For all formulas $\phi \in \lang_K$, for all pointed simplicial models $(\C,X)$ where $X$ is a facet: $\C,X \bowtie\phi$ iff $\kappa(\C,X) \bowtie \phi$, and $\C,X \models\phi$ iff $\kappa(\C,X) \models \phi$.
\end{proposition}
\begin{proof}
The proof is by induction on $\phi$. Recall that $\kappa(\C,X) = (\kappa(\C),X)$.

\medskip
\noindent $
\C,X\bowtie p_a \quad \text{ iff }\\
a \in \chi(X) \quad \text{ iff (*)}\\
X \in \FF(C)_a \quad \text{ iff }\\
\kappa(\C),X \bowtie p_a$

\medskip
\noindent $(*)$: Note that $\FF(C)_a = \{ X \in \FF(C) \mid X \sim_a X\} = \{ X \in \FF(C) \mid a \in \chi(X) \}$.

\medskip
\noindent $
\C,X \bowtie \neg\phi \quad \text{ iff }\\
\C,X \bowtie \phi \quad \text{ iff (induction)}\\
\kappa(\C),X \bowtie \phi \quad \text{ iff }\\
\kappa(\C),X \bowtie \neg\phi$

\medskip
\noindent $
\C,X \bowtie \phi\et\psi \quad \text{ iff }\\
\C,X \bowtie \phi \text{ and } \C,X \bowtie \psi \quad \text{ iff (induction)}\\
\kappa(\C),X \bowtie \phi \text{ and } \kappa(\C),X \bowtie \psi \quad \text{ iff }\\
\kappa(\C),X \bowtie \phi\et\psi$

\medskip
\noindent $
\C,X \bowtie \M_a\phi \quad \text{ iff }\\
\C,Y \bowtie \phi \text{ for some } Y \text{ with } a \in \chi(X\inter Y) \quad \text{ iff ($\Imp$: $Y \subseteq Z$ and Lemma~\ref{lemma.upbowtie}, $\Pmi$: trivial) } \\
\C,Z \bowtie \phi \text{ for some } Z \in \FF(C) \text{ with } a \in \chi(X\inter Z) \quad \text{ iff (induction)} \\
\kappa(\C),Z \bowtie \phi \text{ for some } Z \in \FF(C) \text{ with } X\sim_a Z \quad \text{ iff } \\
\kappa(\C),X \bowtie \M_a \phi$

\medskip

We continue with the satisfaction relation (in the case negation, we use induction for $\bowtie$ and for $\models$, as in the previous proposition).

\medskip
\noindent $
\C,X\models p_a \quad \text{ iff }\\
a \in \chi(X) \text{ and } p_a \in \ell(X) \quad \text{ iff }\\
X \in \FF(C)_a \text{ and } p_a \in L(X) \quad \text{ iff }\\
\kappa(\C),X \models p_a$

\medskip
\noindent $
\C,X \models \neg\phi \quad \text{ iff }\\
\C,X \bowtie \phi \text{ and } \C,X \not\models \phi \quad \text{ iff (twice induction)}\\
\kappa(\C),X \bowtie \phi \text{ and } \kappa(\C),X \not\models \phi  \quad \text{ iff }\\
\kappa(\C),X \models \neg\phi$

\medskip
\noindent $
\C,X \models \phi\et\psi \quad \text{ iff }\\
\C,X \models \phi \text{ and } \C,X \models \psi \quad \text{ iff (induction)}\\
\kappa(\C),X \models \phi \text{ and } \kappa(\C),X \models \psi \quad \text{ iff }\\
\kappa(\C),X \models \phi\et\psi$

\medskip
\noindent $
\C,X \models \M_a\phi \quad \text{ iff }\\
\C,Y \models \phi \text{ for some } Y \text{ with } a \in \chi(X\inter Y) \quad \text{ iff ($\Imp$: $Y \subseteq Z$ and Proposition~\ref{prop.star}; $\Pmi$: trivial)} \\
\C,Z \models \phi \text{ for some } Z \in \FF(C) \text{ with } a \in \chi(X\inter Z) \quad \text{ iff (induction)} \\
\kappa(\C),Z \models \phi \text{ for some } Z \in \FF(C) \text{ with } X\sim_a Z \quad \text{ iff } \\
\kappa(\C),X \models \M_a \phi$
\end{proof}

Propositions~\ref{prop.corr2} and \ref{prop.corr3} established truth value preservation and definability preservation for corresponding simplicial models and local epistemic models. Given that validity is defined in terms of definability and truth, we may conclude that the logics (the set of validities) with respect to both classes of structures are the same. 

\begin{example}[Improper Kripke models] 
The requirement that local epistemic models are proper rules out some models that succinctly describe uncertainty of agents about other agents being alive or dead, but that have the same information content as some other, bigger models. Figure~\ref{figure.lozenge} demonstrates how two agents $a,b$ are uncertain whether a third agent $c$ is alive, and also how the $\sigma$ construction returns an isomorphic local epistemic model. The four-state model, let us call it $\model$, has the same information content as both two-state models in the middle row, that are improper. Whether the value of $p_c$ is true or false when $c$ is dead is irrelevant. That it became false is an artifact of the $\kappa$ construction: it does not label states with variables of agents that are dead, but as these variables are then undefined, their absence does not mean that they are false: $\model, 1_a1_b0_c \not\bowtie p_c$ so $\model, 1_a1_b0_c \not\models\neg p_c$.

Now consider the two-agent Kripke model below in the figure. It is improper. Agent $a$ is only uncertain whether agent $b$ is alive. We do not know how to represent the same information in a simplicial model. As this seems a sad note to end a nice story, let us make it into a more appealing cliffhanger instead. Instead of simplicial complexes that are sets of subsets of a vertex set, consider multisets of subsets. These called \emph{pseudocomplexes} in \cite{hilton_wylie_1960} and \emph{simplicial sets} in \cite{ledent:2019}. One can imagine allowing impure complexes that are such sets of multisets containing multiple copies of a simplex such that one copy is the face of a facet but another copy consisting of the same vertices is a facet in itself. The bottom Kripke model thus corresponds to such a `pseudo impure simplicial model' that consists of an $ab$-edge where $b$ is alive as well as a for $a$ indistinguishable $a$-vertex $v$ that is a facet and where $b$ is dead. Similarly we can represent the other improper models as pseudocomplexes consisting of a single triangle with one of the edges occurring twice.
\end{example}

\begin{figure}[ht]
\center
\scalebox{.75}{
\begin{tikzpicture}[round/.style={circle,fill=white,inner sep=1}]
%\fill[fill=gray!25!white] (0.29,1) -- (2,0) -- (2,2) -- cycle;
\fill[fill=gray!25!white] (4,0) -- (4,2) -- (5.71,1) -- cycle;
\fill[fill=gray!25!white] (0.29,1) -- (2,0) -- (2,2) -- cycle;
\node[round] (b1) at (.29,1) {$1_c$};
\node[round] (b0r) at (5.71,1) {$1_c$};
\node[round] (c1) at (2,2) {$1_b$};
\node[round] (c1r) at (4,2) {$1_a$};
%\node[round] (lc1) at (1,1.71) {$0_c$};
\node[round] (a0) at (2,0) {$1_a$};
\node[round] (a0r) at (4,0) {$1_b$};
%\node (f1) at (4.6,1) {$Y'$};
%\node (f1) at (1.4,1) {$X'$};
%\node(cp) at (-.5,0) {$\C':$};
%
\draw[-] (b1) -- (a0);
\draw[-] (b1) -- (c1);
%\draw[-] (a0) -- (lc1);
\draw[-] (a0r) -- (b0r);
\draw[-] (b0r) -- (c1r);
\draw[-] (a0) -- (c1);
\draw[-] (a0r) -- (c1r);
\draw[-] (a0) -- (a0r);
\draw[-] (c1) -- (c1r);
\end{tikzpicture}
\quad $\stackrel \kappa \Imp$  \quad
\begin{tikzpicture}
\node (lz) at (-2,0.4) {\scriptsize$abc$};
\node (rz) at (1.9,0.4) {\scriptsize$abc$};
\node (tz) at (-.5,1.9) {\scriptsize$ab$};
\node (bz) at (-.5,-1.9) {\scriptsize$ab$};
\node (l) at (-1.5,0) {$1_a1_b1_c$};
\node (r) at (1.5,0) {$1_a1_b1_c$};
\node (t) at (0,1.5) {$1_a1_b0_c$};
\node (b) at (0,-1.5) {$1_a1_b0_c$};
\draw[-] (l) -- node[above left] {$b$} (t);
\draw[-] (l) -- node[below left] {$a$} (b);
\draw[-] (t) -- node[above right] {$a$} (r);
\draw[-] (b) -- node[below right] {$b$} (r);
\end{tikzpicture}
\quad $\stackrel\sigma\Imp$ \quad
\begin{tikzpicture}[round/.style={circle,fill=white,inner sep=1}]
%\fill[fill=gray!25!white] (0.29,1) -- (2,0) -- (2,2) -- cycle;
\fill[fill=gray!25!white] (4,0) -- (4,2) -- (5.71,1) -- cycle;
\fill[fill=gray!25!white] (0.29,1) -- (2,0) -- (2,2) -- cycle;
\node[round] (b1) at (.29,1) {$1_c$};
\node[round] (b0r) at (5.71,1) {$1_c$};
\node[round] (c1) at (2,2) {$1_b$};
\node[round] (c1r) at (4,2) {$1_a$};
\node[round] (a0) at (2,0) {$1_a$};
\node[round] (a0r) at (4,0) {$1_b$};
\draw[-] (b1) -- (a0);
\draw[-] (b1) -- (c1);
%\draw[-] (a0) -- (lc1);
\draw[-] (a0r) -- (b0r);
\draw[-] (b0r) -- (c1r);
\draw[-] (a0) -- (c1);
\draw[-] (a0r) -- (c1r);
\draw[-] (a0) -- (a0r);
\draw[-] (c1) -- (c1r);
\node (b1a) at (-.2,1.4) {\scriptsize$(\{1_a1_b1_c\},c)$};
\node (b0ra) at (6.2,1.4) {\scriptsize$(\{1_a1_b1_c\},c)$};
\node (c1a) at (1.3,2.4) {\scriptsize$(\{1_a1_b0_c,1_a1_b1_c\}, b)$};
\node (c1ra) at (4.7,2.4) {\scriptsize$(\{1_a1_b0_c,1_a1_b1_c\},a)$};
\node (a0a) at (1.3,-.4) {\scriptsize$(\{1_a1_b0_c,1_a1_b1_c\}, a)$};
\node (a0ra) at (4.7,-.4) {\scriptsize$(\{1_a1_b0_c,1_a1_b1_c\}, b)$};
\end{tikzpicture}
}

\bigskip
\bigskip

\scalebox{0.75}{
%\hspace{2cm}
\begin{tikzpicture}
\node (010l) at (0,0.4) {\scriptsize$ab$};
\node (001l) at (3,0.4) {\scriptsize$abc$};
\node (010) at (.5,0) {$1_a1_b0_c$};
\node (001) at (3.5,0) {$1_a1_b1_c$};
\draw[-] (010) -- node[above] {$ab$} (001);
\end{tikzpicture}
\qquad
\begin{tikzpicture}
\node (010l) at (0,0.4) {\scriptsize$ab$};
\node (001l) at (3,0.4) {\scriptsize$abc$};
\node (010) at (.5,0) {$1_a1_b1_c$};
\node (001) at (3.5,0) {$1_a1_b1_c$};
\draw[-] (010) -- node[above] {$ab$} (001);
\end{tikzpicture}
}

\bigskip
\bigskip

\scalebox{0.75}{
%\hspace{2cm}
\begin{tikzpicture}
\node (010l) at (0.1,0.4) {\scriptsize$a$};
\node (001l) at (3.1,0.4) {\scriptsize$ab$};
\node (010) at (.5,0) {$1_a1_b$};
\node (001) at (3.5,0) {$1_a1_b$};
\draw[-] (010) -- node[above] {$a$} (001);
\end{tikzpicture}
}

\caption{Different representations for two agents only uncertain whether a third agent is alive. Below, a model for one agent only uncertain whether a second agent is alive.}
\label{figure.lozenge}
\end{figure}
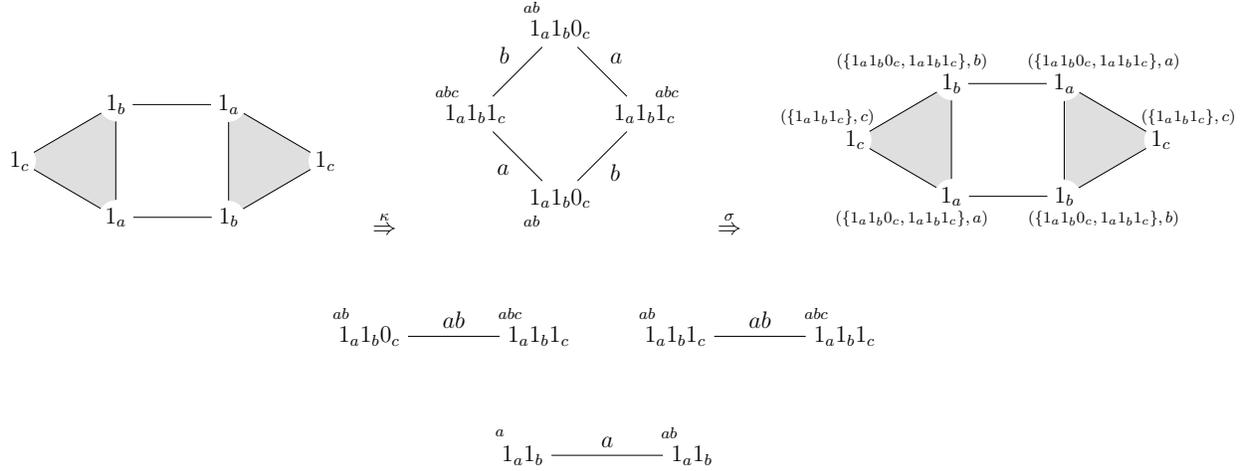

\section{Comparison to the literature} \label{section.further}

We will now discuss in greater detail relations of our proposal to work on awareness of agents, many-valued logics, other notions of knowledge and belief, correctness (of agents in relation to protocols), and temporal modal logic interpreted on simplicial complexes.

\paragraph{Awareness of agents.} Knowing that an agent is alive is like saying that you are aware of that agent. Various (propositional) modal logics combine knowledge and uncertainty with awareness and unawareness \cite{faginetal:1988,halpernR13,AgotnesA14a,hvdetal.jolli:2014,DitmarschFVW18}. However, in all those except \cite{hvdetal.jolli:2014} this is awareness of sets of \emph{formulas}, not awareness of \emph{agents}. One could consider defining awareness of agents in a given state as awareness of all formulas involving those agents, as a generalization of the so-called awareness of propositional variables (primitive propositions) of \cite{faginetal:1988}. Given a logical language $\lang$ for variables $P$, this means awareness of all formulas in the language $\lang|Q$ for some $Q \subseteq P$. Now if, in our setting, $Q = \Union_{a \in B} P_a$ for some $B \subseteq A$, then awareness of fragment $\lang|(\Union_{a\in B} P_a)$ means unawareness of any local variables of agents not in $B$. This is somewhat as if the agents in $B$ are alive and those not in $B$ are dead. However, for definability it not merely counts what agent's variables $p_a$ occur in a formula, but also what agent's modalities $K_a$ occur in a formula. That goes beyond straightforward \cite{faginetal:1988}-style awareness of propositional variables. And even if that were possible, it would be different from our essentially modal setting: whether a formula is defined in a given state is not a function of what agents are alive in that state, but a function of what agents are alive in indistinguishable states (and so on, recursively). Still, in a simplex wherein all the agents occurring in a formula $\phi$ are alive, the formula is defined: we recall  Lemma~\ref{lemma.aphibowtie} stating exactly that: $A_\phi \subseteq \chi(X)$ implies $\C,X \bowtie \phi$.  

\paragraph{Many-valued logic.}
Our semantics is a three-valued modal logic with a propositional basis known as Kleene's weak three-valued logic: the value of any binary connective is unknown if one of it arguments is unknown \cite{sep-logic-manyvalued}. Modal logical (including epistemic) extensions of many-valued logics are found in \cite{Morikawa89,Fitting:1992,odintsovetal:2010,RivieccioJJ17} (there is no relation to embeddings of three-valued propositional logics into modal logics \cite{KooiT13}). This sounds promising, however, before a link can be made with our work there are a lot of caveats. In the three-valued logics of this crowd the third value stands for `both true and false', and the four-valued (so-called Belnapian) logics of this crowd add to this a fourth value with the intended meaning `neither true nor false', in other words `unknown'. These are paraconsistent logics. Unknown and undefined are similar but not the same: a disjunction $p \vel q$ is true if one of the disjuncts is true and the other is unknown, but a disjunction $p \vel q$ is undefined if one of the disjuncts is true and the other is undefined. More importantly, in such many-valued modal logics the modal extension tends to be independent from the many-valued propositional base. But not in our case: whether a formula is defined depends on the modalities occurring in it, not only on the propositional variables.

\paragraph{Belief and knowledge.} The logic {\bf S5}$^\top$ of knowledge on impure complexes is `almost' {\bf S5}. Is it yet another epistemic notion that is `almost' like knowledge and hovering somewhere between belief and knowledge? It is not. It is not Hintikka's favourite {\bf S4} \cite{hintikka:1962} nor {\bf S4.4}  \cite{stalnaker:2005,HalpernSS09a}, as axiom {\bf 5} ($\M_a\phi \imp K_a \M_a\phi$) is valid (Proposition~\ref{prop.standard}). It is not {\bf KD45} \cite{handbookintro:2015} either, so-called `consistent belief', as {\bf T} ($K_a \phi \imp \phi$) is valid (Proposition~\ref{prop.standard}). However, let us recall the peculiarity of {\bf T} in our setting: that $K_a \phi \imp \phi$ is valid does {\bf not} mean that if $K_a\phi$ is true then $\phi$ is true. It means that if $K_a\phi$ is true and $\phi$ is defined then $\phi$ is true. Which is the same as saying that if $K_a\phi$ is true then $\phi$ is not false. This seems to come close to the motivation of `belief as defeasible knowledge' in \cite{MosesS93}. Let us consider that.

We recall $B^\alpha \phi$ of \cite{MosesS93}: the agent believes $\phi$ on \emph{assumption} $\alpha$. Could this assumption not be that `$\phi$ is defined'? One option they consider for $B^\alpha \phi$ is to define this as $K(\alpha \imp \phi) \et \M \phi$  \cite[Def.\ 3, page~304]{MosesS93}. This appears to come close to our derived semantics of knowledge $K \phi$ (Proposition~\ref{lemma.diamond}) as, in the incarnation for Kripke models: `in all indistinguishable states, whenever $\phi$ is defined, it is true, and there is an indistinguishable state where $\phi$ is true'. But, although coming close, it is not the same. Their results are for assumptions $\alpha$ that are \emph{formulas}, and in a binary-valued (true/false) semantics. Whereas our definability assumption is a feature grounding a three-valued semantics. Still, our motivation is very much that of \cite{MosesS93} as well as \cite{stalnaker:2005,HalpernSS09a}: how to define belief from knowledge, rather then knowledge from belief.

Semantics of knowledge for impure complexes were also considered in \cite{diego:2019} and in \cite{goubaultetal:2021}. 

In \cite{diego:2019} it is observed that impure simplicial models correspond to Kripke models where dead agents (crashed processes) do not have equivalence accessibility relations, and as this is undesirable for reasoning about knowledge, projections are proposed from impure complexes to pure subcomplexes for the subset of agents that are alive \cite[Section 3.3]{diego:2019}, and from impure complexes to Kripke models for the subset of agents that are alive \cite[Section 3.4]{diego:2019}. One could imagine a logic based on this observation but this would then not allow agents that are alive to reason about agents that are dead, which seems restrictive. 

In \cite{goubaultetal:2021}, the authors propose an epistemic semantics for impure complexes for which they also show correspondence with (certain) Kripke models with symmetric and transitive relations. There are interesting similarities and differences with our approach. In \cite{goubaultetal:2021}, the Kripke/simplicial correspondence is shown between Kripke frames (without valuations $L$) and chromatic simplicial complexes (without valuations $\ell$), not between Kripke models and simplicial models. Basically, the construction is the same as ours.\footnote{We acknowledge \'Eric Goubault kindly sharing older unpublished work on this matter.} The differences appear when we go from frames to models, namely $(i)$ in their \emph{different notion of valuation on complexes}, and $(ii)$ in their \emph{different knowledge semantics}. To illustrate that, we recall once more the impure complex $\C$ from Example~\ref{example.xxx} (page~\pageref{example.xxx}) consisting of an edge $X$ and a triangle $Y$ intersecting in an $a$-vertex. First, in \cite{goubaultetal:2021}, valuations $\ell$ assign variables not to vertices but to facets. We therefore have a choice whether $\ell(X) = \{p_b\}$ so that $p_b$ is true and $p_a$ and $p_c$ are false, or $\ell'(X) = \{p_b,p_c\}$ so that $p_b$ and $p_c$ are true and $p_a$ is false. In our semantics $p_c$ is undefined in $X$ and therefore neither true nor false. Second, in \cite{goubaultetal:2021}, definability is not an issue, any formula $\phi$ is always defined, and an agent knows the proposition expressed by $\phi$, if $\phi$ is true in all facets containing that agent's vertex. Formulas are always defined because local variables always have a value, also in facets of smaller than maximal dimension, like the edge $X$ in $\C$. There are therefore two different impure complexes like $\C$, one with $\ell$ where $K_a p_c$ is false (as $p_c$ is false in $X$), and another one with $\ell'$ where $K_a p_c$ is true (as $p_c$ is true in $X$ and in $Y$). In our semantics, for $K_a p_c$ to be true, it suffices that $p_c$ is true in triangle $Y$. This is consistent with, say, edge $X$ having been a triangle where $p_c$ was true before $c$ crashed, and also consistent with edge $X$ having been a triangle where $p_c$ was false before $c$ crashed. The difference between the respective semantics becomes notable when applied to the `benchmark' modelling Example~\ref{figure.sergio} from \cite{herlihyetal:2013}: for $a$ to know that the value of $c$ is $1$, variable $p_c$ must be true on edges lacking a $c$ vertex. This reflects that agent $a$ already received confirmation of $c$'s value before becoming uncertain whether $c$ subsequently crashed.\footnote{We acknowledge several discussions on this matter with Sergio Rajsbaum and with Ulrich Schmid.} Maybe, such modelling desiderata can also be explicitly realized with a (event/action) history-based semantics for complexes just as for Kripke models in \cite{jfaketal.JPL:2009}. 

The logic for impure complexes of \cite{goubaultetal:2021} is axiomatized by the modal logic {\bf KB4} (where {\bf B} is the axiom $\phi \imp K_a \M_a \phi$). The logic {\bf S5}{+}{\bf L} of \cite{goubaultetal:2018} is now not the special case for impure complexes, as in Section~\ref{section.pure}: because valuations assign variables to facets and not to vertices, the locality axiom {\bf L} ($K_a p_a \vel K_a \neg p_a$) is invalid for the semantics of \cite{goubaultetal:2021}.

The semantics of \cite{goubaultetal:2021} can explicitly treat life and death in the language. If $a$ is dead then $K_a \bot$ is true (where $\bot$ is the always false proposition) and if $a$ is alive then $\neg K_a \bot$ is true. Alive agents have correct beliefs: $\neg K_a \bot \imp (K_a \phi \imp \phi)$ is valid. These are appealing properties. We recall that in our semantics we cannot formalize that agents are alive or dead. We consider this a feature rather than a problem. In the next section, in the paragraph `Death in the language' we show that adding global variables for `agent $a$ is alive' to the language (with appropriate semantics) makes our logic more expressive. 

It seems of interest to determine the expressivity of a version of \cite{goubaultetal:2021} with a primitive modality meaning $\neg K_a \bot \et K_a \phi$, for `what alive agents know', that combines existential and universal features, somewhat more akin to our knowledge semantics and to the `hope' modality in \cite{abs-2106-11499,KuznetsP0F19} modelling correctness, discussed below.
 
\paragraph{Alive and dead versus correct and incorrect.}
Instead of being alive or dead (crashing), processes may merely be correct or incorrect, for example as a consequence of unreceived messages or Byzantine behaviour \cite{abs-2106-11499,KuznetsP0F19}. We anticipate that certain impure complexes might be equally useful to model knowledge and uncertainty about correct/incorrect agents analogously to modelling knowledge and uncertainty about alive/dead agents.

\paragraph{A temporal modal logic on simplicial complexes.} 
In \cite{loretietal:2023} a temporal modal logic with `next time' and `until' modalities is presented  that is interpreted on non-chromatic simplicial complexes, in a line of research on spatial modal logics  related to epistemic logics interpreted on topological structures \cite{DabrowskiMP96,jfaketal.space:2004,parikhetal:2007}. The facets of their complexes can have different dimensions just as in the impure complexes of our contribution. The vertices of their complexes are not decorated with agents or with local variables (they are not chromatic and not local). However, the authors envisage an epistemic topological interpretation wherein the set of vertices represents a set of (all different) agents. Under that interpretation, and using their notion called spatial adjacency (non-empty intersection of faces), the semantics of `next time' correspond to the semantics of `somebody knows' \cite{AgotnesW21}, and the semantics of `until' correspond to that of relativized common knowledge \cite{jfaketal.lcc:2006,hvdetal.del:2007}. 

\section{Further research} \label{section.fff}

\paragraph*{Completeness of S5$^\top$.}  
The completeness of an extension of the axiomatization {\bf S5}$^\top$ is shown in \cite{rojo:2023}. It involves the construction of a canonical simplicial model for the three-valued semantics.

\paragraph*{Bisimulation.}
Bisimulation between pure complexes is presented in \cite{ledent:2019,hvdetal.simpl:2022,goubaultetal_postdali:2021}. (Bisimulation between certain non-chromatic complexes is also presented in \cite{loretietal:2023}.) A generalization of this notion of bisimulation to one between impure complexes (and analogously between local epistemic models) is proposed in \cite{Ditmarsch21}. Unfortunately it is faulty and it has therefore been removed from this extended version. It may be of interest that we sketch some issues there. Consider the two simplicial models below.

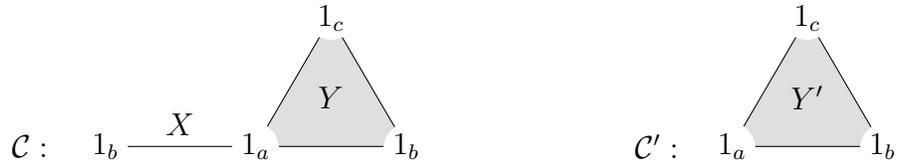
\begin{center}
\begin{tikzpicture}[round/.style={circle,fill=white,inner sep=1}]
%\fill[fill=gray!25!white] (0,0) -- (2,0) -- (1,1.71) -- cycle;
\fill[fill=gray!25!white] (2,0) -- (4,0) -- (3,1.71) -- cycle;
\node[round] (b1) at (0,0) {$1_b$};
\node[round] (b0) at (4,0) {$1_b$};
\node[round] (c1) at (3,1.71) {$1_c$};
\node[round] (a0) at (2,0) {$1_a$};
%\node (bb1) at (0,-.4) {$z$};
%\node (bb0) at (4,-.4) {$u$};
%\node (ba0) at (2,-.4) {$w$};
%\node (c10) at (3,2.11) {$x$};
\node (f1) at (3,.65) {$Y$};
\draw[-] (b1) -- node[above] {$X$} (a0);
\node (cc) at (-1,0) {$\C:$};%
\draw[-] (a0) -- (b0);
\draw[-] (b0) -- (c1);
\draw[-] (a0) -- (c1);
\end{tikzpicture}
\qquad\qquad\qquad
\begin{tikzpicture}[round/.style={circle,fill=white,inner sep=1}]
\fill[fill=gray!25!white] (2,0) -- (4,0) -- (3,1.71) -- cycle;
\node[round] (b0) at (4,0) {$1_b$};
\node[round] (c1) at (3,1.71) {$1_c$};
\node[round] (a0) at (2,0) {$1_a$};
%\node (bb0) at (4,-.4) {$u'$};
%\node (ba0) at (2,-.4) {$w'$};
%\node (c10) at (3,2.11) {$x'$};
\node (f1) at (3,.65) {$Y'$};
\node (cc) at (1,0) {$\C':$};%
\draw[-] (a0) -- (b0);
\draw[-] (b0) -- (c1);
\draw[-] (a0) -- (c1);
\end{tikzpicture}
\end{center}
We claim that the pointed models $(\C,Y)$ and $(\C',Y')$ have the same information content with respect to our language and semantics. One can show that $(\C,Y)$ and $(\C',Y')$ are {\em modally equivalent} in the sense that for all $\phi \in \lang_K$: $\C,Y \bowtie \phi$ iff $\C',Y' \bowtie \phi$, $\C,Y \models \phi$ iff $\C',Y' \models \phi$, and $\C,Y \models \neg \phi$ iff  $\C',Y' \models\neg \phi$. This is the obvious sense in our three-valued semantics: whatever of the three values, they should correspond. %Let us denote such modal equivalence as $(\C,Y)\equiv(\C',Y')$. 
Here, it may be surprising that the $X$ edge in $\C$ does not play havoc: whatever agent $a$ knows about $c$ is witnessed by the facet $Y$, we do not `need' $X$ for that. Therefore it can be missed in $\C'$ without pain. It is harder to come up with a notion of bisimulation such that $(\C,Y)$ is bisimilar to $(\C',Y')$. %
%, denoted as $(\C,Y)\bisim(\C',Y')$. 
The issue is what to do with the $X$ facet in $\C$ when doing a {\bf forth} step for agent $a$ in $Y$. We note that $\M_b 1_c$ is undefined in $(\C,X)$ and that there is therefore no edge in $\C'$ that is modally equivalent to $X$, as this requires being equidefinable. It therefore seems that there should also be no edge in $\C'$ that is bisimilar to $X$. One can resort to ad-hoc solutions for this particular example. Why not say that $\C$ is a simulation of $\C'$, instead? But it is easy to come up with more complex structures that both contain different impurities such that neither is a simulation of of the other.
Clearly we want the Hennessy--Milner property that bisimilarity corresponds to modal equivalence.

\paragraph*{Death in the language.}
We cannot formalize `$a$ knows that $b$ is dead' or `$a$ knows that $b$ is alive' or even `$a$ considers it possible that $b$ is dead/alive' in the current logical language. This was on purpose: we targeted the simplest epistemic logic. But it is quite possible. To the set of local variables $P = \Union_{a \in A} P_a$ we add a set $A_\downarrow := \{a_\downarrow \mid a \in A\}$ of {\em global variables} $a_\downarrow$ denoting ``$a$ is alive''. We further define by abbreviation $a_\uparrow := \neg a_\downarrow$, for ``$a$ is dead''. The upward pointing arrow is to suggest that dead agents go to heaven. Let now a simplicial model $\C = (C,\chi,\ell)$ be given. First, we require for $v \in \VV(C)$ that $a_\downarrow \in \ell(v)$ iff $\chi(v) = a$. This still seems obvious, as agent $a$ colouring $v$ is alive. This also entails that $a_\downarrow \in \ell(Y)$ for any $Y \in C$ with $v \in Y$. Next, for any \emph{facet} $X \in \FF(C)$ we define $\C,X \bowtie a_\downarrow$ (that is, always). Otherwise, $a_\downarrow$ is undefined. And we then stipulate for such $X$ that 
$\C,X \models a_\downarrow$ iff $\C,X \bowtie a_\downarrow$ and $a_\downarrow \in \ell(X)$. Under these circumstances, a validity is $K_a a_\downarrow$, formalizing that an agent knows that it is alive. 

The precise impact on the logic is unclear to us. Do we still have upwards and downwards monotony? Is it sufficient to add axioms $K_a a_\downarrow$ to the axiomatization {\bf S5}$^\top$? It has a big impact on expressivity. The two models $(\C,Y)$ and $(\C',Y')$ in the prior discussion on bisimulation, that were indistinguishable in our semantics, can now be distinguished by $\M_a c_\uparrow$: on the left, agent $a$ considers it possible that agent $c$ is dead, but on the right, she knows that $c$ is alive.

\paragraph*{Impure simplicial action models.}
It is straightforward to model updates on simplicial complexes as \emph{action model} execution \cite{baltagetal:1998}, generalizing the results for pure complexes of  \cite{goubaultetal:2018,goubaultetal_postdali:2021,ledent:2019,diego:2019,diego:2021}.  We envisage (possibly impure) \emph{simplicial action models} representing incompletely specified tasks and algorithms, for example Byzantine agreement \cite{DworkM90} on dynamically evolving  (with agents dying) impure complexes. It should be noted however that our framework does not enjoy the property that positive formulas (intuitively, those without negations before $K_a$ operators; corresponding to the universal fragment in first-order logic) are preserved after update. It might therefore be challenging to generalize results for pure complexes employing such `positive knowledge gain' after update \cite{goubaultetal:2018,goubaultetal_postdali:2021} to truly failure-prone distributed systems modelled with impure complexes.

\paragraph*{Distributed knowledge and common knowledge.}
Various notions of group knowledge can enrich the logical language, such as mutual knowledge, distributed knowledge, and common knowledge (see \cite{handbookintro:2015} for an overview) but also more (currently) exotic notions such as `somebody knows' \cite{AgotnesW21} and common distributed knowledge  \cite{vanwijk:2015,Baltag20,hvdetal.simpl:2022}. Semantics of common knowledge and (common) distributed knowledge on pure simplicial models have been presented in \cite{ledent:2019,hvdetal.simpl:2022}. An axiomatization involving distributed knowledge on a further generalization of impure simplicial models corresponding to so-called simplicial sets has been proposed in \cite{goubaultetal:2023}. Combining group epistemics with dynamics such as action models may lead to further challenges for axiomatizations, and seems promising to address connectivity problems and results in distributed computing.

%\section{Conclusion} \label{section.conclusion}

\section*{Acknowledgements} This work is the extended version of \cite{Ditmarsch21}. We acknowledge interactions and comments from: Armando Casta\~{n}eda, J\'er\'emy Ledent, \'Eric Goubault, Yoram Moses, Sergio Rajsbaum, Rojo Randrianomentsoa, David Rosenblueth, Ulrich Schmid, and Diego Vel\'azquez. Diego's exploration of the semantics of knowledge on impure complexes \cite{diego:2019} inspired this investigation. Hans warmly recalls his stay at UNAM on Sergio's invitation where his adventures in combinatorial topology for distributed computing got started, and his stay at TU Wien on Ulrich's invitation where these adventures were continued. We thank the reviewers for their comments and their pointers to the literature. Roman Kuznets is supported by the Austrian Science Fund (FWF) project ByzDEL (P 33600).

\bibliographystyle{plain}
\bibliography{biblio2023}

\end{document}